\documentclass[11pt,a4wide]{amsart}

\usepackage{graphicx}
\usepackage{psfrag} 
\usepackage{mathptmx}      
\usepackage{amssymb}
\usepackage{amsmath}
\usepackage{color}
\usepackage{natbib}

\newtheorem{lemma}{Lemma}[section]
\newtheorem{theorem}[lemma]{Theorem}
\newtheorem{proposition}[lemma]{Proposition}
\newtheorem{corollary}[lemma]{Corollary}

\newcommand{\cL}{{\mathcal L}}

\newcommand{\cA}{{\mathcal A}}

\DeclareMathOperator*{\Block}{Block}

\newcommand{\iV}{\mathring{V}}

\title{Distinguished minimal toplogical lassos}

\author{Katharina T.\ Huber\address{School of Computing Sciences,
 University of East Anglia, Norwich, NR4 7TJ, UK}
\and
        George Kettleborough\address{School of Computing Sciences,
 University of East Anglia, Norwich, NR4 7TJ, UK}
}

\date{\today}

\begin{document}

\maketitle

\begin{abstract}
A classical result in distance based tree-reconstruction characterizes 
when for a distance $D$ on some finite set $X$ there exist a  uniquely
determined dendrogram on $X$ (essentially a rooted tree $T=(V,E)$ with 
leaf set $X$ and no degree two vertices but possibly the root and an 
edge weighting $\omega:E\to \mathbb R_{\geq 0}$) such 
that the distance $D_{(T,\omega)}$
induced by $(T,\omega)$ on $X$ is $D$.  Moreover,
algorithms that quickly reconstruct 
 $(T,\omega)$ from $D$ in this case are known.
However in many 
areas where dendrograms are being constructed such as Computational Biology
not all distances on $X$ are always available implying that
 the sought after dendrogram 
need not be uniquely determined anymore by the available distances
with regards to topology of the underlying tree, edge-weighting, or both. 
To better understand the 
structural properties a set $\cL\subseteq {X\choose 2}$ has to 
satisfy to overcome this problem,
various types of lassos have been introduced. 
Here, we focus on the question of when
a lasso  uniquely determines the topology of 
a dendrogram's underlying
 tree, that is, it is a topological lasso for that tree.
We show that any
set-inclusion minimal topological lasso for such a tree 
$T$ can be transformed into a 'distinguished' 
minimal topological lasso $\cL$ for $T$, that is, 
the graph $(X,\cL)$ is a claw-free block graph. Furthermore, we
characterize such lassos in terms of the novel concept
of a cluster marker map for $T$ and present results concerning 
the heritability of such lassos in the context of the subtree 
and supertree problems.
\end{abstract}

 \noindent {\bf Keywords:}
dendrogram, block graph, claw free, topological lasso, 
$X$-tree\\

\noindent{ \bf AMS:}
 05C05, 92D15

\pagestyle{myheadings}
\thispagestyle{plain}
\markboth{K.\,T.\,Huber AND G.\,Kettleborough}{Distinguished minimal 
topological lassos}

\section{Introduction}

In many topical studies in Computational Biology ranging from 
gene onthology via genome wide association studies in population genetics
to evolutionary genomics,
the following fundamental mathematical problem is encountered: 
Given a distance $D$
on some set $X$ of objects, find a dendrogram $\mathcal D$ on $X$ (essentially
a rooted tree $T=(V,E)$ with no degree two vertices but
possibly the root whose leaf 
set is $X$ together with an edge-weighting 
$\omega:E\to\mathbb R_{\geq 0}$
-- see Fig.~\ref{fig:block-graph-motivation}
for  examples) such that the distance induced by $\mathcal D$
on any two of its leaves $x$ and $y$ 
equals $D(x,y)$. In the ideal case that the distances between 
any two elements of $X$ are available, it is well-understood
when such a tree is uniquely determined by them and fast algorithms
for reconstructing it from them are known 
(see e.\,g.\,\cite[Chapter 9.2]{DHKMS11}
and \cite[Chapter 7.2]{SS03}  where dendrograms 
are considered in the slightly more general forms of dated rooted $X$-trees
and equidistant representations of dissimilarities, respectively,
and \cite[Chapter 3]{BG91} 
as well as the references in all three of them for more on this).
 
The reality however tends to be different in many cases in that
distances between pairs of objects might be 
missing or are  not sufficiently reliable to warrant 
inclusion of that distance in 
an analysis -- see e.g. \cite{P04,SMS10,SS10} for more on this topic
in an evolutionary genomics context). 
 Exclusion of such a distance might therefore be tempting but is clearly not 
always desirable which raises 
interesting mathematical, statistical, and algorithmical questions
(see e.\,g.\,\cite{DS84,F95,SG92} for a study concerning the latter
and \cite{F95,GG99,G04,M01} for results concerning its unrooted variant).
One of them is the focus of this paper: Calling 
any subset of a finite set $X$ of size two a {\em cord} of  $X$ 
then for what sets $\cL$ of cords of $X$ do we need to know the
distances so that both the topology of the underlying tree and the edge-weights
of the dendrogram on $X$ that induced the distances on the cords
in $\cL$ is uniquely determined by $\cL$?

To help illustrate the intricacies of this question which is concerned
with the structure of the set $\cL$ and not so much with 
the actual distances on the cords in $\cL$,
denote for any two distinct elements $a,b\in X$ the cord $\{a,b\}$
by $ab$. Consider the dendrogram $\mathcal D$ with leaf set
$X=\{a,\ldots, e\}$ depicted in
Fig.~\ref{fig:lasso-example}(i) and assume
 that the distances on the cords of $\cL=\{ac,de,bc,ce,cd\}$
are induced by $\mathcal D$ so, for example, the distance 
on the cord $ab$ 
is four. Then the dendrogram $\mathcal D'$
depicted in Fig.~\ref{fig:lasso-example}(ii) induces the
same distances on the cords in $\cL$ as $\mathcal D$
but the topologies of the underlying trees $T$ and $T'$
of $\mathcal D$ and $\mathcal D'$, respectively, 
are clearly not the same in the sense that there exists no
bijection from $V(T)$ to $V(T')$ that is the
identity on $\{a,\ldots, e\}$ and induces a  rooted graph 
isomorphism from  $T$ to $T'$.
Thus, $\cL$ does
not uniquely determine $T$ and thus also not $\mathcal D$. However
as can be quickly checked the situation changes if and only if
the cord $ab$ (or a subset of ${X\choose 2}$ containing that cord)
is added to $\cL$.  To make this more precise, 
let $\cL'$ denote the resulting set of cords on $X$
and let $\mathcal D_1$ denote a dendrogram on $X$ for which the
topology of the underlying tree 
is the same as that of $\mathcal D$. If $\mathcal D_2$ is a dendrogram
on $X$  such that the distances on the
cords in $\cL'$ induced by $\mathcal D_1$ and $\mathcal D_2$ coincide then,
as is easy to verify, the topologies of the 
underlying trees of $\mathcal D_1$ and $\mathcal D_2$,  
respectively, must be the same 
and so must be their edge-weightings. Thus, $\cL'$
uniquely determines $\mathcal D$. 
\begin{figure}[h]
\begin{center}
\input{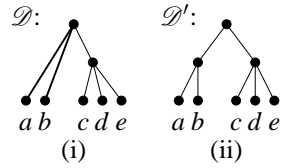}
\end{center}
\caption{\label{fig:lasso-example}
For $X=\{a,\ldots, e\}$ and $\cL=\{ac,de,bc,ce,cd\}$
the dendrogram $\mathcal D$  and $\mathcal D'$ are depicted in 
(i) and (ii), respectively. 
Bold edges
in $\mathcal D$ have weight two and all other edges as well 
as all edges in $\mathcal D'$
have weight one.
}
\end{figure}

Although an intriguing question, apart  from some recent results in \cite{HP13},
not much is known about it (see \cite{DHS11} and \cite{HS13}
for some partial results in case the tree in question is unrooted).
By formalizing a dendrogram in terms of a certain  edge-weighted
$X$-tree (see the next section for a precise definition
of this concept as well as all the other concepts
mentioned below) and using the concept of a
topological lasso which was originally introduced 
for unrooted phylogenetic trees with leaf set $X$
in \cite{DHS11} 
and extended to $X$-trees in \cite {HP13}, we study this question in
the form of when a set of cords of $X$  is a topological
lasso for a given $X$-tree $T$. In the context of this,
we are particularly interested in (set-inclusion) minimal topological
lassos $\cL$ for $T$ for which $
\bigcup \cL:=\bigcup_{A\in\cL} A=X$ holds. 

For $T$ an $X$-tree, we show for any such minimal topological
lasso $\cL$ for $T$ that in case the graph $\Gamma(\cL)$ whose
vertex set is $X$ and any two distinct elements $x$ and $y$
in $X$ joined by an edge if $xy\in \cL$ -- see
Fig~\ref{fig:block-graph-motivation}(i) for an example of
that graph for $\cL=\{ab,cd,ef,ac,ce,ea\}$ -- is a block
graph then the blocks of $\Gamma(\cL)$ are in 
one-to-one correspondence with
the non-leaf vertices of $T$ (Corollary~\ref{cor:bijection}). 
Furthermore, we establish in Theorem~\ref{theo:transform}
that any minimal topological lasso $\cL$ for $T$ can be transformed
into a very special type of minimal topological lasso $\cL^*$
for $T$ in that $\Gamma(\cL^*)$
is a claw-free block graph where a graph is called
{\em claw-free} if it does not contain a {\em claw}, that is, 
the complete bipartite graph $K_{1,3}$ as an induced subgraph \cite{H72}.
Claw-free graphs have been shown to enjoy numerous properties relating
them to, for example, perfect graphs, perfect matchings, 
and maximum independent sets
(see e.\,g.\,\cite{FFZ97} and \cite{CFHV12} for overviews).  
Furthermore, claw-free block graphs were related in
\cite{BL09} to $k$-leaf powers of trees and their spectrum
was studied in \cite{GS01, MSST06} (see also
\cite{BR13} for a more general study of the adjacency matrix
of such graphs). 
Calling a minimal topological lasso $\cL$ for $T$ {\em distinguished}
if $\Gamma(\cL)$ is a claw-free block graph,
we present in Theorem~\ref{theo:characterization}
a characterization of a distinguished
minimal topological lasso for $T$  in terms of the novel
concept of a cluster marker
map for $T$. In addition, we characterize when a
distinguished minimal topological lasso for $T$ 
gives rise to a distinguished minimal topological
lasso for a subtree of $T$ (Theorem~\ref{theo:subtree})
and also present a partial answer to the canonical
analogue of a question raised for supertrees of unrooted
phylogenetic trees in \cite{DHS11}.

The paper is organized as follows. In Section~\ref{sec:terminology},
we introduce relevant terminology surrounding $X$-trees and
lassos. In Section~\ref{sec:gamma-l-graph}, we 
collect first properties of the graph $\Gamma(\cL)$ associated
to a topological lasso $\cL$ and in Section~\ref{sec:blockgraph}, we 
establish Corollary~\ref{cor:bijection}. In Section~\ref{sec:distinguished}, 
we commence our study of distinguished minimal topological
and establish Theorem~\ref{theo:transform}. In Section~\ref{sec:sufficient},
we present a sufficient condition for when a 
minimal topological lasso is distinguished 
(Theorem~\ref{theo: distinguished-lasso-verification}) and in
Section~\ref{sec:characterization-distinguished}, we 
prove Theorem~\ref{theo:characterization}. We conclude with
Section~\ref{sec:subtree} where we establish Theorem~\ref{theo:subtree}
and also outline directions for further research.

\section{Basic terminology and assumptions}\label{sec:terminology}
In this section, we introduce some relevant 
basic terminology surrounding $X$-trees,
their edge-weighted counterparts, and lassos.  Assume throughout the paper  
that $X$ is a finite set with at least 3 elements and that, unless
stated otherwise, all sets $\cL$ of
cords of $X$ considered in this paper 
satisfy the property  that
$X=\bigcup \cL$.

\subsection{$X$-trees}
A {\em rooted tree} $T$ is a tree with a unique distinguished vertex 
called the
{\em root} of $T$, denoted by $\rho_T$. Throughout the paper, we 
assume that the degree of the root of a rooted tree is at least two.
A {\em rooted phylogenetic $X$-tree}, or {\em $X$-tree} for short, 
is a rooted tree $T=(V,E)$ with no degree two vertices but
possibly the root $\rho_T$ whose leaf set is $X$. We call 
an $X$-tree $T$ a {\em star-tree
on $X$} if every leaf of $T$ is adjacent with the root  of $T$ 

Suppose for the following  that
$T$ is an $X$-tree. Then we call a vertex of $T$ that is not
 a leaf of $T$ an {\em interior vertex}
of $T$ and denote the set of interior vertices of $T$ by $\iV(T)$.
We call an edge of $T$  that is incident with a leaf of $T$ a {\em pendant
edge} of $T$ and every edge of $T$ that is not a pendant edge
an {\em interior edge} of $T$.
Extending some of the terminology for directed graphs
to $X$-trees, we call for all vertices $v\in V(T)-\{\rho_T\}$
 an edge $e\in E(T)$
a  {\em parent edge of $v$} if $e$ is incident with $v$ and 
lies on the path from 
the root $\rho_T$ of $T$ to $v$. We refer to the vertex incident with $e$ 
 but distinct from $v$ as a {\em parent} of $v$.
 
Suppose for the following
that $v$ is an interior vertex of $T$. If $v$ is not the root of $T$
then we call an edge $e\in E(T)$ a {\em child edge of $v$} if $e$ is
incident with $v$ but is not crossed by the path from $\rho_T$ to $v$.  
In addition, we
call every edge incident with $\rho_T$ a {\em child edge of $\rho_T$}.
We call the vertex incident with a child edge of an
 interior vertex $w$ of $T$
 but distinct from $w$ a {\em child of $w$} and denote
 the set of all children of $v$ by $ch_T(v)$.
We call a vertex $w\in V(T)$ distinct from $v$
a {\em descendant} of $v$ if either $w$ is a child of $v$ or
there exists a path from $v$ to $w$ that crosses a child of $v$.
We denote the set of leaves 
of $T$ that are also descendants of $v$ by $L_T(v)$. If $v$ is a leaf
of $T$ then we put $L_T(v):=\{v\}$. If there
is no ambiguity as to which $X$-tree $T$ we are referring to then,
for all $v\in V(T)$,  we will write $L(v)$ rather than $L_T(v)$
and $ch(v)$ rather than $ch_T(v)$.  

We call a non-empty subset $L\subsetneq X$ of 
leaves of $T$ such that $L=L(v)$ holds for some
$v\in \iV(T)$ a {\em pseudo-cherry} of $T$. In that case,
we also call $v$ the {\em parent} of that pseudo-cherry. 
Note that every 
$X$-tree on three or more leaves must
contain at least one pseudo-cherry. Also note
that a pseudo-cherry of size two is 
a {\em cherry} in the usual sense (see e.g. \cite{SS03}).

For $x$ and $y$  distinct elements in $X$,
we call the unique vertex of $T$ that simultaneously lies on
the path from $x$ to $y$, on the path from $x$ to $\rho_T$, and
on the path from $y$ to $\rho_T$ the {\em last common ancestor of
$x$ and $y$}, denoted by $lca_T(x,y)$. More generally, for any 
subset $Y\subseteq X$ of size three or more,
we denote the subtree of $T$ with leaf set $Y$ and
vertices of degree two suppressed (except possibly the root)
by $T|_Y$ and call the root of $T|_Y$ the {\em last common ancestor of
$Y$}, denoted by $lca_T(Y)$.

Finally, suppose that $T'$ is be a further
$X$-tree. Then we say that $T$ and $T'$ 
are {\em equivalent} if there exists a bijection $\phi:V(T)\to V(T')$
that extends to a graph isomorphism between $T$ and $T'$ that is
the identity on $X$ and maps the root $\rho_T$ of $T$ to the
root $\rho_{T'}$ of $T'$. 

\subsection{Edge-weighted $X$-trees and lassos}

Suppose for the following again that $T$ is an $X$-tree.
An {\em edge weighting $\omega$ of $T$} is a map
$\omega :E(T)\to \mathbb R_{\geq 0}$ that maps every
edge of $T$ to a non-negative real. Suppose that $\omega$
is an edge-weighting for $T$. Then
we call the pair $(T,\omega)$ an {\em edge-weighted} $X$-tree and $\omega$
{\em proper} if $\omega(e)>0$ holds for every interior edge
$e$ of $T$. We denote the distance induced by $(T,\omega)$
on the leaves of $T$ by $D_{(T,\omega)}$
and call $\omega$ {\em equidistant} if 
\begin{enumerate}
\item[(i)] $D_{(T,\omega)}(x,\rho_T)=  D_{(T,\omega)}(y,\rho_T)$, for all
$x,y\in X$, and
\item[(ii)] $D_{(T,\omega)}(x,u)\geq  D_{(T,\omega)}(x,v)$, for all
$x\in X$ and all $u,v\in V$ such that $u$ is encountered
before $v$ on the path from $\rho_T$ to $x$.
\end{enumerate} 

Suppose $\cL$ is a set of cords of $X$. Then  
we call two edge-weighted $X$-trees 
$(T_1,\omega_1)$ and $(T_2,\omega_2)$  {\em $\cL$-isometric} if 
$D_{(T_1,\omega_1)}(x,y)=D_{(T_2,\omega_2)}(x,y)$ 
holds for all cords $xy\in \cL$. We say that $\cL$ is
a  {\em topological lasso} for $T$ if for every $X$-tree $T'$ and
any equidistant, proper edge-weightings $\omega$ of $T$ and $\omega'$ of
$T$',  we have that $T$ and $T'$ are equivalent
whenever $(T,\omega)$ and  $(T',\omega')$ are $\cL$-isometric.
If $\cL$ is a topological lasso for $T$ then we also say that $T$ is 
{\em topologically lassoed} by $\cL$. Moreover, 
we say that $\cL$ is  a {\em (set-inclusion) minimal topological 
lasso for $T$} if 
$\cL$ is a topological lasso for $T$ but no cord $A\in \cL$
can be removed from $\cL$ such that $\cL-\{A\}$ is still a
topological lasso for $T$. 
For ease of readability, if the $X$-tree to which a
topological lasso $\cL$ refers is of no relevance
to the discussion, we will simply say that $\cL$ is a 
topological lasso. 

To illustrate some of these definitions, let $X=\{a,\ldots, f\}$ and
let $\cL$ be the set of cords such that $\Gamma(\cL)$ is the
graph depicted in Fig.~\ref{fig:block-graph-motivation}(i).
Using e.\,g.\,\cite[Theorem~7.1]{HP13} (see also 
Theorem~\ref{theo:characterization-topology} below)
it is easy to see that the $X$-trees depicted in
Fig.~\ref{fig:block-graph-motivation}(ii) and (iii) respectively
are topologically lassoed by $\cL$. In fact, $\cL$ is a minimal 
topological lasso for both of them.
\begin{figure}[h]
\begin{center}
\input{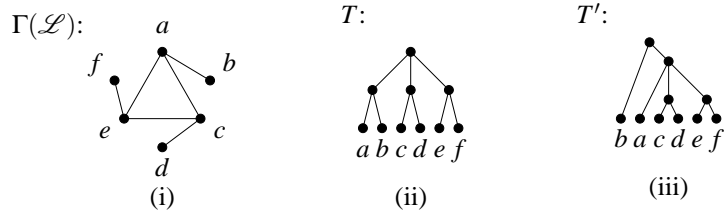}
\end{center}
\caption{\label{fig:block-graph-motivation}
(i) The graph $\Gamma(\cL)$ with vertex set $X=\{a,b,\ldots,f\}$
for the set $\cL=\{ab,cd,ef,ac,ce,ea\}$. (ii) and (iii)
Two non-equivalent $X$-trees $T$ and $T'$
that are both topologically lassoed by $\cL$. In fact, 
$\cL$ is a minimal topological lasso for either one of them.
}
\end{figure}

\section{The graphs $\Gamma(\cL)$ and 
$G(\cL,v)$}\label{sec:gamma-l-graph}

In this section, we investigate properties of the graph $\Gamma(\cL)$ 
associated to a 
set $\cL$ of cords of $X$. We start by remarking  that
if there is no danger of confusion, we denote an edge 
$\{a,b\}$ of $\Gamma(\cL)$ by $ab$ rather than $\{a,b\}$.

To establish our first structural result 
for $\Gamma(\cL)$ (see 
Proposition~\ref{prop:gamma-l-connected}), 
we require further terminology. 
Suppose $T$ is an $X$-tree, $v\in \iV(T)$, and 
$\cL$ is a 
set of cords of $X$. Then we call
the graph $G_T(\cL,v)=(V_{T,v},E_{T,v})$ with vertex set $V_{T,v}$
the set of all child edges of $v$ and edge set $E_{T,v}$ the
set of all $\{e,e'\}\in {V_{T,v}\choose 2}$ for which there
exist leaves $a,b\in X$ such that
$e$ and $e'$ are edges on the path from $a$ to $b$ in $T$ 
and $ab\in \cL$ holds the {\em child-edge graph of $v$ (with respect 
to $T$ and $\cL$)}. Note that in case there is no danger of 
ambiguity with regards to the $X$-tree $T$ we are referring to, we
 will write $G(\cL,v)$ rather than $G_T(\cL,v)$ and $V_v$ and $E_v$
rather than $V_{T,v}$ and $E_{T,v}$. The next result which was originally
established in \cite[Theorem 7.1]{HP13}
states a a crucial property of child-edge graphs.

\begin{theorem}\label{theo:characterization-topology}
Suppose $T$ is an $X$-tree and 
$\cL$ is a set of cords of $X$. 
Then the following are equivalent:
\begin{enumerate}
\item[(i)] $\cL$ is a topological lasso for $T$.
\item[(ii)] for every vertex $v\in \iV(T)$, the graph
$G(\cL,v)$ is a clique.
\end{enumerate}
\end{theorem}

Denoting for an $X$-tree $T$, a topological lasso $\cL$ for
$T$, and an interior vertex $v\in \iV(T)$ 
the set of all cords $ab\in \cL$ for which
$v=lca_T(a,b)$ holds  by 
$\cA(v)$, Theorem~\ref{theo:characterization-topology}
readily implies $|\cA(v)|\geq {|ch(v)|\choose 2}$.
The next observation is almost trivial yet central to the paper and  
concerns the special case 
that $\cL$ is a minimal topological lasso for $T$. Its
proof which combines a straightforward counting argument 
with Theorem~\ref{theo:characterization-topology} is left to
the interested reader. To able to state it,
we denote for an interior vertex $v\in \iV(T)$ and a child edge $e\in E(T)$
of $v$  the child of $v$ indicent with $e$ by $v_e$.  

\begin{lemma}\label{lem:size-A(v)}
Suppose $T$ is an $X$-tree and $\cL$ is a minimal topological 
lasso for $T$. Then, for all $v\in \iV(T)$, we have 
$|\cA(v)|={|ch(v)|\choose 2}$. In particular, for any two distinct child 
edges $e_1$ and 
$e_2$ of $v$ there exists precisely one pair 
$(a_1,a_2)\in L(v_{e_1})\times L(v_{e_2})$ such that 
$a_1a_2\in\cL$.
\end{lemma}

Note that Lemma~\ref{lem:size-A(v)} immediately implies that
any two minimal topological lassos for the same $X$-tree must
be of equal size. 

To be able to establish Proposition~\ref{prop:gamma-l-connected}, 
we require a further definition.
Suppose $T$ is an $X$-tree and $\cL$ is a topological 
lasso for $T$. Then for all $v\in V(T)$, we denote by $\Gamma_v(\cL)$
the subgraph of $\Gamma(\cL)$ induced by $L(v)$. Note that
in case $v$ is a leaf of $T$ and thus an element in $X$
the only vertex in $\Gamma_v(\cL)$ is $v$ (and 
$E(\Gamma_v(\cL))=\emptyset$).

\begin{proposition}~\label{prop:gamma-l-connected}
Suppose $T$ is an $X$-tree and $\cL$ is a topological lasso for $T$. 
Then, for all $v\in V(T)$, the graph $\Gamma_v(\cL)$ is connected. 
In particular, 
$\Gamma(\cL)$ is connected.
\end{proposition}
\begin{proof}
Assume for contradiction that there exists some vertex $v\in V(T)$
such that  $\Gamma_v(\cL)$ is not connected. Then $v$ cannot be
a leaf of $T$ and so $v\in \iV(T)$ must hold. Without loss of generality
we may assume that $v$ is such that for all descendants $w\in V(T)$ of $v$
the induced graph $\Gamma_w(\cL)$ is connected. Since $\cL$
is  a topological lasso for $T$ and so $G(\cL,v)$ is a clique,
it follows for any two distinct children $v_1,v_2\in ch(v)$  that 
there exists a pair $(x_1,x_2)\in L(v_1)\times L(v_2)$
such that $x_1x_2\in \cL$.  Since the assumption on $v$ implies 
that the graphs 
$\Gamma_{w}(\cL)$ are connected for all children $w\in ch(v)$,   
it follows that $\Gamma_v(\cL)$ is connected which is impossible.
Thus, $\Gamma_v(\cL)$ is connected, for all  $v\in V(T)$.
That $\Gamma(\cL)$ is connected is a trivial consequence.
\qquad
\end{proof}

\section{The case that $\Gamma(\cL)$ is a 
block graph}\label{sec:blockgraph}

To establish a further property of $\Gamma(\cL)$ 
which we will do in Proposition~\ref{prop:x-i-unique}, we require some
terminology related to block graphs (see e.\,g.\,\cite{diestel}). 
Suppose $G$ is a graph. Then a vertex of $G$ 
is called a {\em cut vertex} if its deletion (plus its
incident edges) disconnects $G$. A graph is called a {\em block} 
if it has at least one vertex, is connected, and does 
not contain a cut vertex. A {\em block of a graph $G$} is a maximal connected
subgraph of $G$ that is a block and a graph is called a {\em block graph}
if all of its blocks are cliques. For convenience, we refer to a
block graph with vertex set $X$ as a {\em block graph on X}.

As the example of the two minimal topological lassos 
$
\{ab,cd,ef,ac,ce,ea\}$ and $
\{ab, bc,$  $cd, de, ef, fa\}$ 
for the $\{a,\ldots,f\}$-tree
depicted in Fig.~\ref{fig:block-graph-motivation}(ii)
indicates, the graph $\Gamma(\cL)$  associated to a
 minimal topological lasso  $\cL$
may be but need not be a block graph. 
However if it is
then Lemma~\ref{lem:size-A(v)}
can be strengthened to the following central result
where for all positive integers $n$ we 
put $\langle n\rangle :=\{1,\ldots, n\}$ and set
$\langle 0\rangle:=\emptyset$.

\begin{proposition}\label{prop:x-i-unique}
Suppose $T$ is an $X$-tree and $\cL$ is a minimal topological 
lasso for $T$ such that $\Gamma(\cL)$ is a block graph.
Let $v\in \iV(T)$ and let $v_1,\ldots, v_l\in V(T)$ denote the
children of $v$ where
$l=|ch(v)|$.
Then, for all $i\in \langle l\rangle$, there exists a unique 
leaf $x_i\in L(v_i)$ such that $x_sx_t\in \cL$ holds for
all $s,t\in\langle l\rangle$ distinct.
\end{proposition}
\begin{proof}
For all $v\in \iV(T)$ and all $w\in ch(v)$, put
$$
L^v_w:=\{x\in L(w): \mbox{ there exist } w'\in ch(v)-\{w\}
\mbox{ and } y\in L(w')
\mbox{ such that } xy\in \cL  \}.
$$
We need to show that $|L^v_w|=1$ holds for all $v\in \iV(T)$ 
and all $w\in ch(v)$. To see this,
note first that since $G(\cL,v)$ is a clique for all $v\in \iV(T)$, 
we have, for all $w\in ch(v)$ with $v\in \iV(T)$, that $L^v_w\not=\emptyset$. 
Thus,  $|L^v_w|\geq 1$ holds for all such $v$ and $w$.

To establish equality, suppose there exists some interior vertex  
$v\in \iV(T)$ and 
some child $v_1\in ch(v)$ such that $|L^v_{v_1}|\geq 2$.
Choose two distinct leaves $x_1$ and $y_1$ of $T$ 
contained in $L^v_{v_1}$
and denote the parent edge of $v_1$ by $e_1$. Note that $v_1=v_{e_1}$.
Since $y_1\in L^v_{v_1}$, there exists a child edge $e_2$ of $v$
distinct from $e_1$ and some $x_2\in L(v_{e_2})$ such that
$y_1x_2\in\cL$. In view of $x_1\in L^v_{v_1}$, we 
distinguish between the cases that
(i) $x_1z\not\in\cL$ holds for all 
$z\in L(v_{e_2})$ and (ii) 
there exists some $z\in L(v_{e_2})$ such that
$x_1z\in\cL$.

Assume first that Case~(i) holds.
Then since $x_1\in L^v_{v_1}$ there exists 
a further child edge $e_3$ of $v$
and some $y_3\in L(v_{e_3})$ such that $x_1y_3\in\cL$.
Since, by Theorem~\ref{theo:characterization-topology}, 
$G(\cL,v)$ is a clique and so $\{e_2,e_3\}$ is an edge 
in $G(\cL,v)$, there must exist leaves 
$y_2\in L(v_{e_2})$ and 
$x_3\in L(v_{e_3})$ such that 
$y_2x_3\in\cL$. By 
Proposition~\ref{prop:gamma-l-connected}, the graphs 
$\Gamma_{v_{e_i}}(\cL)$, $i=2,3$,
are connected and, by definition, clearly do not share a
vertex. Hence, there must exist a cycle  in $\Gamma(\cL)$ whose vertex
set contains $\bigcup_{j\in\langle 3\rangle} \{x_j,y_j\}$. But then
$x_1x_2\in \cL$ must hold since
 $\Gamma(\cL)$
is a block graph and so every block in $\Gamma(\cL)$ 
is a clique. By Lemma~\ref{lem:size-A(v)} 
applied to $e_1$ and $e_2$, it follows that $x_1=y_1$
as $x_1,y_1\in L(v_1)$ and $y_1x_2\in \cL$ which
is impossible.

Now assume that Case~(ii) holds, that is, there exists some
$z\in L(v_{e_2})$ such that $x_1z\in\cL$. Then 
Lemma~\ref{lem:size-A(v)} applied to $e_1$ and $e_2$
 implies $x_1=y_1$ as $y_1x_2\in\cL$ also holds which is impossible.
 \qquad
\end{proof}

To illustrate Proposition~\ref{prop:x-i-unique}, 
let $T$ be the $X$-tree depicted
in Fig.~\ref{fig:block-graph-motivation}(ii)
and let $\cL$ be the set of cords of $X$ whose $\Gamma(\cL)$ graph is
pictured in  Fig.~\ref{fig:block-graph-motivation}(i). 
Using the notation from Proposition~\ref{prop:x-i-unique} and
 labelling the children of the root of $T$
from left to right by $v_1$, $v_2$ and $v_3$ 
it is easy to see that  Proposition~\ref{prop:x-i-unique} holds for
$x_1=a$, $x_2=c$ and $x_3=e$. 

The next result is the main result of this section and
lies at the heart of Corollary~\ref{cor:bijection}
which provides for an $X$-tree $T$ and a minimal topological lasso
$\cL$ for $T$ such that $\Gamma(\cL)$ is a block graph
a close link between
the blocks of $\Gamma(\cL)$, the interior vertices of $T$ and, for all 
$v\in \iV(T)$, the child-edge graphs $G(\cL,v)$. To establish it, we denote 
for all $v\in V(T)-\{\rho_T\}$ the parent edge of $v$ by $e_v$ and
the set of blocks of
a graph $G$ by $Block(G)$.

\begin{theorem}\label{theo:unique-block}
Suppose $T$ is an $X$-tree and $\cL$ is a minimal topological lasso for
$T$ such that $\Gamma(\cL)$ is a block graph. Then, for all $v\in \iV(T)$,
there exists a unique block $B\in Block(\Gamma(\cL))$ such that
$v=lca_T(V(B))$.
\end{theorem}
\begin{proof}
We first show existence. Suppose $v\in \iV(T)$. Let
$v_1,\ldots,v_l\in V(T)$ denote the children of $v$  where $l=|ch(v)|$.
By Proposition~\ref{prop:x-i-unique}, there exists, for 
all $i\in\langle l\rangle$, a unique
leaf $x_i\in L(v_i)$ such that, for all  $s,t\in \langle l\rangle$ distinct,
we have $x_sx_t\in \cL$. Put $A=\{x_1,\ldots,x_l\}$.
Clearly, $v=lca_T(A)$ and the graph $G(v)$ with vertex set $A$ and
edge set $E=\{\{x,y\}\in {A\choose 2}: xy\in\cL\}$ is a clique.
Then since $\Gamma(\cL)$ is a block graph
there must exist a block $B\in Block(\Gamma(\cL))$ that contains
$G(v)$ as an induced subgraph.

We claim that the graphs $G(v)$ and $B$  are equal. 
In view of the facts that $A\subseteq V(B)$, the
 blocks in a bock graph are cliques, and $G(v)$
is a clique it suffices to show that
$V(B)\subseteq A$. Suppose for contradiction that
there exists some $y\in V(B)-A$. Note first that $yx\in \cL$
must hold for all $x\in A$.
Next note that $y$ cannot be a descendant of $v$
since otherwise there would exist some $i\in\langle l\rangle$ such 
that $y\in L(v_i)$. Choose some  $j\in \langle l\rangle-\{i\}$. Then
Lemma~\ref{lem:size-A(v)} applied to $e_{v_i}$ and $e_{v_j}$
 implies $x_i=y$  as $yx_j,x_ix_j\in\cL$
 which is impossible.

Choose some $x\in A$ and put 
$w=lca_T(x,y)$. Then $v$ is a descendant of $w$ and $w=lca_T(x,y)$ 
holds for all $x\in A$. Let $w_1\in V(T)$ and $w_2\in \iV(T)$ 
denote two distinct children of $w$ such that $y\in L(w_1)$ 
and $x\in L(w_2)$. Then Lemma~\ref{lem:size-A(v)} applied 
to $e_{w_1}$ and $e_{w_2}$
implies $x_i=x_j$ for all $i,j\in\langle l\rangle$ distinct
since $yx\in\cL$ holds
for all $x\in A$ which is impossible.
Thus, $V(B)\subseteq A$, as required. This concludes the proof of
the existence part of the theorem.

We next show uniqueness. 
Suppose for contradiction that there exists
some $v\in \iV(T)$ and distinct blocks $B,B'\in Block(\Gamma(\cL))$
such that  $lca_T(B)=v=lca_T(B')$. 
Since every block of $\Gamma(\cL)$ contains at least two vertices
as $\Gamma(\cL)$ is connected and $|X|\geq 3$,
we may choose distinct vertices  $b_1,b_2\in V(B)$ and 
$b_1',b_2'\in V(B')$ such that 
$lca_T(b_1,b_2)=lca_T(B)=v=lca_T(B')=lca_T(b_1',b_2')$. 
Note that $b_1b_2$ and  $b_1'b_2' $ must be cords in $\cL$
as $B$ and $B'$ are cliques of $\Gamma(\cL)$. 
We distinguish between the
cases that (i) $\{b_1,b_2\}\cap\{b_1',b_2'\}=\emptyset $
and (ii) $\{b_1,b_2\}\cap\{b_1',b_2'\}\not=\emptyset $.

We first show that Case (i) cannot hold. Assume for contradiction that
Case (i) holds, that is, $\{b_1,b_2\}\cap\{b_1',b_2'\}=\emptyset $.
We claim that 
$lca_T(b_1,b_1')=v$. Assume for contradiction that 
$w:=lca_T(b_1,b_1')\not=v$. Let $v_1\in ch(v)$ 
such that $v_1$ lies on the path from $v$ to
$w$. If $v\not =lca_T(b_2,b_2')$
then there exists a descendant $w'\in V(T)$ of $v$ such
that $lca_T(b_2,b_2')=w'$. Let $v_2\in ch(v)$ 
such that $v_2$ that lies on the path from $v$ to $w'$. Then
Lemma~\ref{lem:size-A(v)} applied to $e_{v_1}$ and $e_{v_2}$
implies $b_1=b_1'$ and $b_2=b_2'$ as $b_1b_2,b_1'b_2'\in \cL$
which is impossible. Thus, $lca_T(b_2,b_2')=v$
must hold. Let $v_2,v_2'\in ch(v)$ such that $b_2\in L(v_2)$
and $b_2'\in L(v_2')$. 
Then since $b_1,b_1'\in L(v_1)$ and $b_1b_2, b_1'b_2'\in \cL$, 
Proposition~\ref{prop:x-i-unique}
implies $b_1'=b_1$.
Consequently,
$\{b_1,b_2\}\cap\{b_1',b_2'\}\not=\emptyset $ which is impossible.
Thus, $lca_T(b_2,b_2')=v$
cannot hold and so 
$$
lca_T(b_1,b_1')=v,
$$ 
as claimed.
Swapping the roles of $b_1,b_1'$ and $b_2,b_2'$ in the previous
claim implies that 
 $v= lca_T(b_2,b_2')$ must hold, too. 
 For $i=1,2$ let $v_i,v_i'\in ch(v)$ 
such that $b_i\in L(v_i)$ and $b_i'\in L(v_i')$.
Then, by Lemma~\ref{lem:size-A(v)}, 
there exist pairs $(c,c')\in L(v_1)\times L(v_1')$ and 
$(d,d')\in L(v_2)\times L(v_2')$ such that $cc',dd'\in\cL$. 
Since $(b_1,b_2)\in L(v_1)\times L(v_2)$ and 
$(b_1',b_2')\in L(v_1')\times L(v_2')$ and $b_1b_2,b_1'b_2'\in\cL$,
 Proposition~\ref{prop:x-i-unique} implies 
that $c=b_1$, $b_2=d$, $d'=b_2'$ and $c'=b_1'$. But then
$C$: $c'=b_1',b_2'=d', d=b_2, b_1=c,c'$ is a cycle in $\Gamma(\cL)$.
Since $\Gamma(\cL)$ is a block graph it follows that there must exist
a block $B^C$ in $\Gamma(\cL)$ that contains $C$. 
Since $\{b_1,b_2\}\subseteq V(B^C)\cap V(B)$ and two distinct blocks of
a block graph can share at most one vertex it follows that $B^C$ and 
$B$ must coincide. Since $\{b_1',b_2'\}\subseteq V(B^C)\cap V(B')$
holds too,
similar arguments imply that $B^C$ must also coincide with $B'$.
Thus, $B$ and $B'$ must be equal which is impossible.
Hence Case~(i) cannot hold, as required.

Thus, Case (ii) must hold, that is,  
$\{b_1,b_2\}\cap\{b_1',b_2'\}\not=\emptyset $. Since any 
two distinct blocks in a block graph can share at most one vertex 
it follows that $|\{b_1,b_2\}\cap\{b_1',b_2'\}|=1$.
Without loss of generality we may assume that $b_1=b_1'$. 
We first claim that 
$$
lca_T(b_2,b_2')=v.
$$
Assume to the contrary that $lca_T(b_2,b_2')\not=v$. 
Then there exist distinct children $v_1,v_2 \in ch(v)$ 
such that $b_1\in L(v_1)$ and $b_2,b_2'\in L(v_2)$ hold. 
Since both $b_1b_2$ and
$b_1'b_2'=b_1b_2'$ are cords in $\cL$, Lemma~\ref{lem:size-A(v)}
applied to 
$e_{v_1}$ and $e_{v_2}$
implies $b_2'=b_2$. Hence, $|\{b_1,b_2\}\cap\{b_1',b_2'\}|=2$ which 
is impossible. Thus, $lca_T(b_2,b_2')=v$, as claimed.

Let $v_1,v_2,v_2'\in ch(v)$  such that
$b_1\in L(v_1)$, $b_2\in L(v_2)$, and $b_2'\in L(v_2')$. 
By Lemma~\ref{lem:size-A(v)}, there exist some
$(c,c')\in L(v_2)\times L(v_2')$ such that $cc'\in \cL$.
Since we also have $(b_1,b_2)\in L(v_1)\times L(v_2)$ with
$b_1b_2\in \cL$ holding and $(b_1,b_2')\in L(v_1)\times L(v_2')$ with
$b_2'b_1=b_2'b_1'\in \cL$ holding, Proposition~\ref{prop:x-i-unique}
implies that $b_2=c$ and $b_2'=c'$. Hence, $C$: $b_1=b_1',b_2'=c',
c=b_2,b_1$ is a cylce in $\Gamma(\cL)$ and so similar arguments
as in the corresponding subcase for Case (i) imply that 
$B$ and $B'$ must coincide which is impossible. Thus, 
$lca_T(b_2,b_2')=v$ cannot hold  which concludes the discussion
of Case (ii) and thus the proof of the uniqueness part of the theorem.
\qquad \end{proof}

In view of Theorem~\ref{theo:unique-block}, we denote for $T$ an $X$-tree,
 a minimal topological lasso $\cL$ for $T$ such that $\Gamma(\cL)$
is a block graph, and a vertex $v\in \iV(T)$ 
the unique block $B$ in $\Gamma(\cL)$ for which 
$v=lca_T(V(B))$ holds by 
$B_v^{\cL}$, or simply by
$B_v$ if the set $\cL$ of cords is clear from the context. 
Moreover, we denote for all 
$x\in L(v)$ the child of $v$ on the path from  $v$ to $x$ by $v_x$. 

\begin{corollary}\label{cor:bijection}
Suppose $T$ is an $X$-tree and $\cL$ is a minimal topological lasso for
$T$ such that $\Gamma(\cL)$ is a block graph. Then the map
$$
\psi:\iV(T)\to Block(\Gamma(\cL)) : v\mapsto B_v
$$
is a bijection with inverse map $\psi^{-1}:Block(\Gamma(\cL))\to \iV(T)$: 
$B\mapsto lca_T(V(B))$. Moreover, the map
$$
\chi:Block(\Gamma(\cL)) \to \{G(\cL, v)\,:\, v\in \iV(T)\}
: B\mapsto G(\cL,\psi^{-1}(B))
$$
is bijective and,
 for all $B\in Block(\Gamma(\cL))$, the map
$$
\xi_B:V(B) \to V_{\psi^{-1}(B)} 
: x\mapsto e_{\psi^{-1}(B)_x}
$$ 
induces a graph isomorphism between $B$ and the child-edge
graph $G(\cL, \psi^{-1}(B))$.
\end{corollary}
\begin{proof}
In view of Theorem~\ref{theo:unique-block}, the map $\psi$ is clearly
well-defined and injective. To see that $\psi$ is surjective
let $B\in Block(\Gamma(\cL))$ and put $v_B=lca_T(V(B))$. Clearly,
$v_B\in \iV(T)$. Since 
$B_{v_B}=\psi(v_B)$ is a block in $\Gamma(\cL)$ for which also
$v_B=lca_T(V(B_{v_B}))$
 holds, Theorem~\ref{theo:unique-block} 
 implies that $\psi(v_B)$ and $B$ must coincide.
Consequently, $\psi$ must also be surjective and thus bijective. That
the map $\psi^{-1}$ is as stated is trivial.
Combined with Theorem~\ref{theo:characterization-topology}, 
the bijectivity of the map $\psi$ implies in particular that, 
for all $B\in Block(\Gamma(\cL))$,
the map $\xi_B: V(B) \to V_{\psi^{-1}(B)}  $ 
from $V(B)$ to the vertex set $V_{\psi^{-1}(B)}  $
of the child-edge graph $G(\cL, \psi^{-1}(B))$ induces a graph
isomorphism between $B$ and $G(\cL, \psi^{-1}(B))$.

To see that the map $\chi$ is bijective 
note first that $\chi$ is well-defined
since $\psi^{-1}(B)\in \iV(T)$ holds 
for all blocks $B\in Block(\Gamma(\cL))$. To see that $\chi$ is
injective assume that there exist blocks $B_1,B_2\in Block(\Gamma(\cL))$
such that $\chi(B_1)=\chi(B_2)$ but $B_1$ and $B_2$ are distinct.
Then $\psi^{-1}(B_1)\not= \psi^{-1}(B_2)$ as $\psi$ is a bijection from
$\iV(T)$ to  $Block(\Gamma(\cL)) $.
Combined with the fact that, for all
$B\in Block(\Gamma(\cL))$, the map $\xi_B$  induces a graph
isomorphism between $B$ and $G(\cL, \psi^{-1}(B))$
it follows that 
$\chi(B_1)=G(\cL,\psi^{-1}(B_1))\not=G(\cL,\psi^{-1}(B_2))=\chi(B_2)$ 
which is impossible. Thus, $\chi$ must be injective. Combined with the
fact that
$|Blocks(\Gamma(\cL))|=|\iV(T)|=|\{G(\cL, v)\,:\, v\in \iV(T)\}|$
it follows that $\chi$ must also be surjective and thus bijective.
\qquad \end{proof}

\section{A special type of minimal 
topological lasso} \label{sec:distinguished}

Returning to the example depicted in Fig.~\ref{fig:block-graph-motivation},
it should be noted that, in addition to being a block graph,
 $\Gamma(\cL)$ enjoys a very special
property where $\cL$ is the minimal topological lasso considered in
that example. More precisely, every vertex of $\Gamma(\cL)$
is contained in at most two blocks. 
Put differently, $\Gamma(\cL)$ is a claw-free graph. Motivated by this, we
call a minimal topological lasso $\cL$ {\em distinguished} if
$\Gamma(\cL)$ is a claw-free block graph.  Note that such
block graphs are precisely the
{\em line graphs of (unrooted) trees} where for any graph $G$ the
associated line graph has vertex set $E(G)$ and two vertices
$a,b\in E(G)$ are joined by an edge if $a\cap b\not=\emptyset$ \cite{H72}.  

In this section, we show in Theorem~\ref{theo:transform}
that distinguished minimal topological 
lassos are a very special type of lasso in that for
every $X$-tree $T$ any minimal topological lasso $\cL$ for $T$
can be transformed into a distinguished
minimal topological lasso 
$\cL^*$ for $T$ via a {\em repeated
application} (i.\,e.\,$l\geq 0$ applications) of the rule: 

\begin{enumerate}
\item[(R)] If $xy,yz\in \cL$ and $lca_T(y,z)$ is a descendant of  
$lca_T(x,y)$ in $T$ then delete $xy$ from the edge set of $\Gamma(\cL)$
and add the edge $xz$ to it. 
\end{enumerate}

Before we make this more precise which we will do next, we remark 
that since a topological lasso for a star tree
is in particular a distinguished minimal topological
lasso for it,  we will for this and the next two sections 
restrict our attention to {\em non-degenerate} $X$-trees,
that is, $X$-trees that are not star trees on $X$.

Suppose $T$ is a non-degenerate $X$-tree and $\cL$
is a set of cords of $X$. Let $\iV(T)$ denote
a set of colors and let 
$$
\gamma_{(\cL,T)}:\cL\to \iV(T):\, ab\mapsto lca_T(a,b)
$$
denote an edge coloring of $\Gamma(\cL)$
in terms of the interior vertices of $T$. Note that
if $\cL$ is a topological lasso for $T$
then Theorem~\ref{theo:characterization-topology} implies that
$\gamma_{(\cL,T)}$ is surjective. Returning to Rule (R),
note that a repeated application of that rule to such a set $\cL$ of cords
results in a  set $\cL'$ of cords that is also a 
topological lasso for $T$. Furthermore, note that if $\cL$ is a minimal
topological lasso for $T$ then $\cL'$ is necessarily also a minimal 
topological lasso for $T$. Finally note for all $v\in \iV(T)$ that
$|\gamma_{(\cL,T)}^{-1}(v)|=1$ or  
$|\gamma_{(\cL,T)}^{-1}(v)|\geq 3$ must hold in this case.

\begin{lemma}\label{lem:3-edges-cycle}
Suppose $T$ is a non-degenerate $X$-tree and $\cL$ is a minimal
topological lasso for $T$.  Put $\gamma=\gamma_{(\cL,T)}$ and assume that
$v\in \iV(T)$ such that $|\gamma^{-1}(v)|\geq 3 $. Then for any three
pairwise distinct cords $c_1,c_2,c_3\in \gamma^{-1}(v)$, 
there exists a 
cycle $C_v$ in $\Gamma(\cL)$ such that $c_1,c_2,c_3\in E(C_v)$ and,  
for all $c\in E(C_v)$,
$\gamma(c)$ either equals $v$ or is a descendant of $v$. 
\end{lemma}
\begin{proof}
Let $v\in \iV(T)$ and let $c_1=x_1y_1$, $c_2=x_2y_2$ and
$c_3=x_3y_3$ denote three pairwise distinct cords in $\gamma^{-1}(v)$. 
For all $i\in\langle 3\rangle$,
let $v_i\in ch(v) $ such that $v_i$ lies on the path from $v$ to $x_i$ in $T$
and let $w_i\in ch(v)$ such that $w_i$ lies on the path from $v$ to
 $y_i$ in $T$.
Then, by Lemma~\ref{lem:size-A(v)},  there exists unique pairs
$(s_1,t_1)\in L(v_1)\times L(v_2)$, 
$(s_2,t_2)\in L(w_2)\times L(w_3)$,
and $(s_3,t_3)\in L(w_1)\times L(v_3)$
such that, for all $i\in\langle 3\rangle$, we have $s_it_i\in \cL$.
Since for all such $i$, we also have that $x_i\in L(v_i)$
and $y_i\in L(w_i)$ and, by Proposition~\ref{prop:gamma-l-connected},
the graphs $\Gamma_{v_i}(\cL)$ and $\Gamma_{w_i}(\cL)$
are connected, it follows that there exists a 
cycle $C_v$ in $\Gamma(\cL)$ that contains, for all $i\in\langle 3\rangle$,
the cords $c_i$ and $s_it_i$ in its edge set.

It remains to show that for every edge $c\in E(C_v)$, we have that
$\gamma(c)$ either equals $v$ or is a descendant of $v$. 
Suppose $c\in E(C_v)$.
If there exists some $i\in\langle 3\rangle$ such that $c\in\{c_i,s_it_i\}$
then $\gamma (c)=v$ clearly holds. So assume that
this is not the case. Without loss of generality, we may assume
that $c$ lies on the path $P$ from $x_1$
to $s_1$ in $C_v$ that does not cross $y_1$. Since
$P$ is a subgraph of $\Gamma_{v_1}(\cL)$ and, implied by
Proposition~\ref{prop:gamma-l-connected},
every edge in $\Gamma_{v_1}(\cL)$ is colored via $\gamma$
with a descendant of $v_1$, it follows that
$\gamma(c)$ is a descendant of $v$.
\qquad
\end{proof}

To establish Theorem~\ref{theo:transform}, we require further terminology.
Suppose $T$ is a non-degenerate $X$-tree, $\cL$ is a minimal
topological lasso for $T$, and $v\in \iV(T)$.
Then we denote by $H_{\cL}(v)$ the 
induced
subgraph of $\Gamma(\cL)$ whose vertex set is the
set of all $x\in X$ that are incident with some
cord $c\in \cL$ for which $\gamma_{(\cL,T)}(c)=v$ holds. 
Moreover, 
 we denote the set of cut vertices of a connected block graph
 $G$ by $Cut(G)$ 
and note that in every connected block graph $G$ there must exist a 
vertex that is contained in at most
one block of $G$. This last observation is
central to the proof of Theorem~\ref{theo:transform}(ii).

\begin{theorem}\label{theo:transform}
Suppose $T$ is a non-degenerate $X$-tree and 
$\cL$ is a minimal topological lasso 
for $T$. Then there exists an ordering $\sigma:
v_0, v_1,\ldots, v_k=\rho_T$, $k=|\iV(T)|$,
of $\iV(T)$ such that the following holds:
\begin{enumerate}
\item[(i)]  There exists a sequence $
\cL_{v_0}=\cL,\cL_{v_1},\ldots, \cL^{\dagger}=\cL_{v_{k}}$ 
of minimal topological lassos $\cL_{v_i}$ for $T$, $i\in \langle k\rangle$, 
such that for all such $i$, we have:
\begin{enumerate}
\item[(L1)] 
$\cL_{v_i}$ is obtained from $\cL_{v_{i-1}}$ via a repeated application 
of Rule (R)
and $H_{\cL_{v_i}}(v_i)$ is a maximal clique in $\Gamma(\cL_{v_i})$.
\item[(L2)] For all $j\in\langle i-1\rangle$, 
$H_{\cL_{v_{i}}}(v_j)$ is a maximal clique in $\Gamma(\cL_{v_{i}})$.
\end{enumerate}
In particular, 
$\Gamma(\cL^{\dagger})$ is a block graph.
\item[(ii)]  If 
$\Gamma(\cL)$ is a block graph then  there exists a sequence
$\cL_{v_0}=\cL,\cL_{v_1},\ldots, \cL^*=\cL_{v_{k}}$ 
of minimal topological lassos $\cL_{v_i}$ for $T$, $i\in\langle k\rangle$, 
such that for all such $i$, we have:
\begin{enumerate}
\item[(L1')] 
$\cL_{v_i}$ is obtained from $\cL_{v_{i-1}}$ via a repeated 
application of Rule (R)
and $\Gamma(\cL_{v_i})$ is a block graph.
\item[(L2')]
$\Gamma_{v_i}(\cL_{v_i})$ is a claw-free block graph. 
\end{enumerate}
In particular, $\cL^*$ is a distinguished minimal topological lasso for $T$.
\end{enumerate}
\end{theorem}
\begin{proof}
For all $i\in \langle k\rangle$, put $\cL_i=\cL_{v_{i}}$
and $\gamma_i=\gamma_{(\cL_i,T)}$.  Clearly, 
if $\cL$ is distinguished then the sequences as described in (i) and (ii) 
exist. So assume that $\cL$ is not distinguished.
For all $v\in \iV(T)$, let $l(v)$ denote the length of the path
from the root $\rho_T$ of $T$ to $v$ and put $h=\max_{v\in \iV(T)}\{l(v)\}$.
Note that $h\geq 1$ as $T$ is non-degenerate.
For all $i\in \langle h\rangle$, let
 $V(i)\subseteq \iV(T)$ denote the set of all interior vertices $v$
of $T$ such that $l(v)=i$.
Let $\sigma$ denote an ordering of the vertices in $\iV(T)$ such that the 
vertices in $V(h)$ come first (in any order), then (again in any order)
the vertices in $V(h-1)$ and so on with the last vertex in that ordering 
being $\rho_T$. 

(i) Suppose $v\in \iV(T)$. If $v\in V(h)$ then
we may assume without loss of generality
that $v=v_1$. Then $v_1$ is the parent of a pseudo-cherry of $T$ 
and so Theorem~\ref{theo:characterization-topology} implies that
$H_{\cL}(v_1)$ is a maximal clique in $\Gamma(\cL)$. Thus, 
$\cL_1:=\cL$ is a minimal topological lasso for $T$ that
satisfies Properties~(L1) and (L2).

So assume that $v\not\in V(h)$. Then there exists some 
$|V(h)|< i\leq k$ such that $v=v_i$. Without loss of 
generality, we may assume that $v_i$ is such that,
for all $j\in \langle i-1\rangle$, $\cL_j$ is a minimal 
topological lasso for $T$ that satisfies 
Properties~(L1) and (L2).
If $v_i$ is the parent of a pseudo-cherry of $T$ then similar arguments 
as before imply that $\cL_i:=\cL_{i-1}$ is a minimal topological
lasso for $T$ that satisfies Properties~(L1) and (L2). So assume
that $v_i$ is not the parent of a pseudo-cherry of $T$.
We distinguish between the cases that 
$H_{\cL_{i-1}}(v)$ is a maximal clique in $\cL_{i-1}$
and that it is not. 

Assume first that $H_{\cL_{i-1}}(v)$ is a maximal clique in 
$\cL_{i-1}$. Then since 
$\cL_{i-1}$ is a minimal topological lasso for $T$ that
 satisfies Properties~(L1) and (L2), it is easy to see that
$\cL_{i}:=\cL_{i-1}$ is also a minimal topological lasso for $T$ that
satisfies Properties~(L1) and (L2). 
So assume that 
$H_{\cL_{i-1}}(v)$ is not a maximal clique in $\cL_{i-1}$.
Then $H_{\cL_{i-1}}(v)$ must contain three pairwise distinct edges, 
$e_1=x_1y_1$, $e_2=x_2y_2$, and $e_3=x_3y_3$ say, such that 
$\{e_1, e_2,e_3\}$ is not the edge set of a $3$-clique in $H_{\cL_{i-1}}(v)$.
For all $i\in\langle 3\rangle $, put $z_i=lca_T(x_i,y_i)$.  
Then Lemma~\ref{lem:3-edges-cycle} combined with a repeated
application of Rule (R) to $\cL_{i-1}$
implies that, for all $i\in\langle 3\rangle$, we can 
find elements $x_i'\in L(z_i)$
such that 
$$
\cL_{i-1}'=\cL_{i-1}-\{x_1y_1,x_2y_2,x_3y_3\}\cup 
\{x_1'x_2', x_2'x_3',x_3'x_1'\}
$$
is a minimal topological lasso for $T$ and the cords
$x_1'x_2'$, $x_2'x_3'$, and $x_3'x_1'$ form a $3$-clique 
in $H_{\cL_{i-1}'}(v)$.
Transforming $\cL_{i-1}'$ further by processing any 
three pairwise distinct edges in $H_{\cL_{i-1}'}(v)$ that 
do not already form a $3$-clique in the same way 
and so on eventually yields a minimal topological lasso 
$\cL_{i}$ for $T$ such that any three pairwise distinct edges in
 $H_{\cL_{i}}(v)$ form a $3$-clique. But this implies that $H_{\cL_{i}}(v)$ is
a maximal clique in $\Gamma(\cL_{i})$
and so Property~(L1) is satisfied by $\cL_i$. Since only edges 
$e$ of $\Gamma(\cL_{i-1})$ 
have been modified by the above transformation 
for which $\gamma_{i-1}(e)=v$ holds and, by assumption,  $\cL_{i-1}$
satisfies Property~(L2) it follows that
$\cL_{i}$ also satisfies that property. 

Processing the successor of $v_i$ in $\sigma$ in the same way and so
on yields a minimal topological lasso 
$\cL^{\dagger}$ for $T$ for which $\Gamma(\cL^{\dagger})$ is a 
block graph. This completes the proof of (i).

(ii) For all $i\in \langle k\rangle$
and all vertices $w\in \iV(T)$ put $B^i_w=B^{\cL_i}_w$.
Suppose that $v\in \iV(T)$. If $v\in V(h)$ then
we may assume without loss of generality
that $v=v_1$. Then $v$ is the parent of a pseudo-cherry of $T$ 
and so $\cL_1:=\cL$ clearly satisfies Properties~(L1')
and (L2').

So assume that $v\not\in V(h)$. Then there exists some 
$|V(h)|< i\leq k$ such that $v=v_i$. Without loss of generality, 
we may assume that $v_i$ is minimal, 
that is, for all $j\in\langle i-1\rangle$, we have that
$\cL_j$ is a minimal topological lasso for $T$ that
satisfies Properties~(L1') and (L2'). If $v$ is the 
parent of a pseudo-cherry of $T$ then similar
arguments as before imply that $\cL_i:=\cL_{i-1}$
satisfies Properties~(L1') and (L2'). So assume that
$v$ is not the parent of a pseudo-cherry of $T$.
If $\Gamma_v(\cL_{i-1})$ is a claw-free  block graph  then
setting $\cL_i:=\cL_{i-1}$ implies that $\cL_i$ satisfies
Properties~(L1') and (L2').

So assume that 
this is not the case, that is, there exists a vertex $x\in L(v)$
that, in addition to being  a vertex in the block 
$B_v^{i-1}$ of $\Gamma(\cL_{i-1})$
 and thus  of $\Gamma_v(\cL_{i-1})$,
is also a vertex in $l\geq 2$ further blocks $B_1,\ldots, B_l$
 of $\Gamma_v(\cL_{i-1})$ which are also blocks in 
$\Gamma(\cL)$. Then there exists a path $P$ from $v$ 
 to $x$ in $T$ that contains, for all $l\geq 2$, the vertices 
 $\psi^{-1}(B_1),\ldots, \psi^{-1}(B_l)$ in its vertex set
where $\psi:\iV(T)\to \Block(\Gamma(\cL))$ is the map from 
Corollary~\ref{cor:bijection}. Let $w\in ch(v)$
 denote the child of $v$ that lies on $P$. Note that since $l\geq 2$, 
we have $w\in \iV(T)$. Without loss of generality, 
 we may assume that $w=v_{i-1}$. 
 The fact that  $\Gamma(\cL_{i-1})$ is a block graph and so
  $\Gamma_{v_{i-1}}(\cL_{i-1})$ is a block graph combined with
the fact that  $\Gamma_{v_{i-1}}(\cL_{i-1})$ is connected 
 implies, in view of the observation
preceding Theorem~\ref{theo:transform}, that
we may choose some
 $y\in L(v_{i-1})-Cut(\Gamma_{v_{i-1}}(\cL_{i-1}))$. Then $y$
 is a vertex in precisely one block of $\Gamma_{v_{i-1}}(\cL_{i-1})$
 and thus can be a vertex in at most two 
 blocks of $\Gamma_v(\cL_{i-1})$.
Consequently, $y\not=x$. 
Applying Rule (R) repeatedly to $\cL_{i-1}$, 
let $\cL_i$ denote the set of cords obtained
from $\cL_{i-1}$ by replacing, for all $i\leq l\leq k$, every cord 
of $\cL_{i-1}$ of the form $xa$ with $a\in V(B_{v_l}^{i-1})$ by the cord
$ya$. Then, by construction, $\cL_i$ is a minimal topological lasso
for $T$ and $\Gamma(\cL_i)$ is a block graph. Hence, $\cL_i$ satisfies
Property~(L1'). Moreover, since
$\Gamma_{v_{i-1}}(\cL_{i-1})$ is claw-free it follows that
$\Gamma_{v_i}(\cL_i)$ is claw-free and so $\cL_i$ satisfies
Property~(L2'), too.

Applying the above arguments to the successor of $v_i$ 
in $\sigma$ and so on eventually
yields a minimal topological lasso $\cL_k$ for $T$  
that satisfies Properties~(L1') and (L2'). Thus,
$\Gamma_{v_k}(\cL_k) $ is a claw-free block graph
and, so, $\cL^*$ is a distinguished
minimal topological lasso for $T$.
\qquad
\end{proof}

To illustrate Theorem~\ref{theo:transform}, let 
$X=\{a,\ldots, f\}$ and consider the
$X$-tree $T'$ depicted in Fig.~\ref{fig:block-graph-motivation}(iii) 
along with the
set $\cL=\{ad,ec,fa,ef,cd,bd\}$ of cords of $X$ which we depict in
Fig.~\ref{fig:transformation}(i) in the form of $\Gamma(\cL)$.
\begin{figure}[h]
\begin{center}
\input{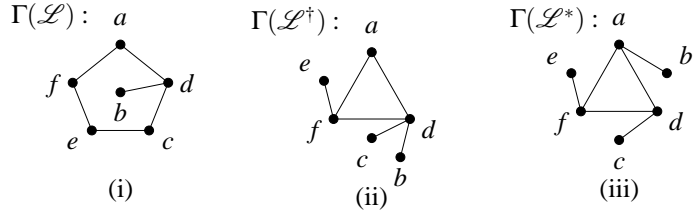}
\end{center}
\caption{\label{fig:transformation}
For $X=\{a,\ldots, f\}$ 
and the $X$-tree $T'$ pictured in Fig.~\ref{fig:block-graph-motivation}(iii),
we depict in (i) the minimal topological lasso $\cL=\{ad,ec,fa,fe,cd,bd\}$ 
for $T'$  in the form of $\Gamma(\cL)$. 
In the same way as in (i), we depict in (ii) 
the transformed minimal topological lasso $\cL^{\dagger}$ for $T'$ such that
$\Gamma(\cL^{\dagger})$  is a block graph and in  (iii) the
distinguished minimal topological lasso $\cL^*$ for $T'$ obtained from
$\cL^{\dagger}$ -- see text for details. 
}
\end{figure}
Using for example Theorem~\ref{theo:characterization-topology},
it is straight-forward to check that $\cL$ is a minimal
topological lasso for $T'$ but $\Gamma(\cL)$ is clearly not a block graph 
and so $\cL$ is also not distinguished. To transform $\cL$ into
a distinguished minimal topological lasso $\cL^*$
for $T'$ as described in Theorem~\ref{theo:transform}, consider the ordering 
$v_1=lca_{T'}(e,f)$, $v_2=lca_{T'}(c,d)$, $v_3=lca_{T'}(a,d)$, $v_4=\rho_{T'}$
of the interior vertices of $T'$. For all $i\in\langle 4\rangle$, 
put $\cL_i=\cL_{v_i}$. Then we first transform $\cL$ into a minimal
topological lasso $\cL^{\dagger}$ for $T'$ as described in 
Theorem~\ref{theo:transform}(i). For this we
have $\cL=\cL_0=\cL_1=\cL_2$
and $\cL_3$ is obtained from $\cL_2$ by first applying Rule (R) to 
the cords $ec, cd \in \cL_2$ resulting in the deletion of the
cord $ce$ from $\cL_2 $ and the addition of the cord $ed$ to $\cL_2$
and then to the cords $fe,ed\in\cL_2$  resulting in the 
deletion of the cord $ed$ from $\cL_2$ and the addition of the cord
$fd$ to it. The graph $\Gamma(\cL_3)$ is  depicted in 
Fig.~\ref{fig:transformation}(ii). 
Note that $\cL_3=\cL^{\dagger}$ and that although $\Gamma(\cL^{\dagger})$
is clearly a block graph $\cL^{\dagger}$ is not distinguished. 

To transform
$\cL^{\dagger}$ into a distinguished minimal topological lasso $\cL^*$ 
for $T'$, we next apply Theorem~\ref{theo:transform}(ii). For this, we
need only consider the vertex $d$ of $\Gamma(\cL^{\dagger})$ that is, we 
have $\cL^{\dagger}=\cL_0=\cL_1=\cL_2=\cL_3$.
Since the child of $v_4$ on the path from $v_4$ to $d$ is 
$v_3$, we may choose $a$ as the element $y$
in $L(v_3)-Cut(\Gamma_{v_3}(\cL_3))$. Then applying Rule (R)
to the cords $bd,da\in \cL_3$ implies the deletion of $bd$ from $\cL_3$
and the addition of the cord $ab$ to it. The resulting minimal topological lasso
for $T'$ is $\cL^*$ which we depict in Fig.~\ref{fig:transformation}(iii)
in the form of $\Gamma(\cL^*)$.
 
We conclude this section by remarking in passing that 
combined with Theorem~\ref{theo:characterization-topology} 
which implies that any minimum sized topological lasso for an $X$-tree
$T$ must have $\sum_{v\in \iV(T)}{|ch(v)|\choose 2}$ cords, 
Theorem~\ref{theo:transform} and Corollary~\ref{cor:bijection}
imply that the minimum sized topological lassos of an $X$-tree $T$
are precisely the minimal topological lassos of $T$.
 
\section{A sufficient condition for a minimal topological lasso
to be distinguished} \label{sec:sufficient}
In this section, we turn our attention towards
presenting a sufficient condition
for a minimal topological lasso for some $X$-tree $T$ to
be a distinguished minimal topological lasso for $T$.
In the next section, we will show that this condition is also 
sufficient.  

We start our discussion with introducing some more terminology.
Suppose $T$ is a non-degenerate $X$-tree. Put 
$cl(T)=\{L(v): v\in \iV(T)-\{\rho_T\}\}$ and note that 
$cl(T)\not=\emptyset$. For all $A\in cl(T)$,
put $cl_A(T):=\{B\in cl(T): B\subsetneq A\}$ and note that
a vertex $v\in \iV(T)-\{\rho_T\}$ is the parent of a pseudo-cherry
of $T$ if and only if $cl_{L(v)}(T)=\emptyset$.
For $\sigma$ 
a total ordering of $X$ and $\min_{\sigma}(C)$ denoting the minimal
element of a non-empty subset $C$ of $X$, we call
a map of the form 
$$
f:cl(T)\to X:
A\mapsto \left\{\begin{array}{cc}
\min_{\sigma}(A-\{f(B): B\in cl_A(T)\})
 & \mbox{ if }cl_A(T)\not=\emptyset,\\
\min_{\sigma}(A)  & \mbox{ else. }
\end{array}
\right.
$$ 
a {\em cluster marker map (for $T$ and $\sigma$)}. 
Note that since $|\iV(T')|\leq |X|-1$ holds for all $X$-trees $T'$
and so $A-\{f(B): B\in cl_A(T)\}\not=\emptyset$  must
hold for all $A\in cl(T)$ with $cl_A(T)\not=\emptyset $,
it follows that  $f$ is well-defined.
Also note that if 
$v\in \iV(T)$ is the parent of a pseudo-cherry $C$ of $T$
then $f(L(v))=f(C)= \min_{\sigma}(C)$
as $cl_C(T)=\emptyset$ in this case. Finally, note that
it is easy to see that a cluster marker map must be injective
but need not be surjective.

We are now ready to present a construction of a distinguished minimal
topological lasso which underpins the aforementioned sufficient condition
that a minimal topological lasso must satisfy to be distinguished.
Suppose that $T$ is a non-degenerate $X$-tree, that $\sigma$ is a 
total ordering of $X$, and that
$f:cl(T)\to X$ is a cluster marker map for $T$
and $\sigma$. We first
associate to every interior vertex $v\in \iV(T)$ a set $\cL_{(T,f)} (v)$
defined as follows. Let $l_1,\ldots, l_{k_v}$ denote the children of $v$
that are leaves of $T$ and let $v_1,\ldots v_{p_v}$ denote the
children of $v$ that are also interior vertices of $T$. Note that $k_v=0$
or $p_v=0$ might hold but not both. Put 
${\emptyset \choose 2}={\langle 1\rangle \choose 2}=\emptyset$. Then we set
$$
\cL_{(T,f)}(v):=\bigcup_{\{i,j\}\in {\langle k_v\rangle\choose 2}}\{l_il_j\}
\cup
\bigcup_{\{i,j\}\in {\langle p_v\rangle\choose 2}}\{f(L(v_i))f(L(v_j))\}
\cup
\bigcup_{i\in \langle k_v\rangle,\,\,j\in \langle p_v\rangle} 
\{l_if(L(v_j))\}.
$$
Note that 
$|\cL_{(T,f)}(v)|\geq 1$
must hold for all $v\in \iV(T)$. Finally, we set
$$
\cL_{(T,f)}:=\bigcup_{v\in \iV(T)} \cL_{(T,f)}(v).
$$

To illustrate these definitions, consider the $X=\{a,\ldots, f\}$-tree 
$T'$ depicted in Fig.~\ref{fig:block-graph-motivation}(iii). 
Let $\sigma$ denote the
lexicographic ordering of the elements in $X$. Then
the map $f:cl(T')\to X$ defined by setting
$$
f(\{c,d\})=c, \,\,\,
f(\{e,f\})=e,\,\,\mbox{ and } f(X-\{b\})=a
$$
is a cluster marker map for $T'$ and $\sigma$ and 
$\cL_{(T,f)}$
(or more precisely the graph $\Gamma(\cL_{(T',f)})$) is depicted in
Fig.~\ref{fig:block-graph-motivation}(i).

To help establish Theorem~\ref{theo: distinguished-lasso-verification}, 
we require
some intermediate results which are of interest in their own right and
which we present next. To this end, we denote
for a vertex $v\in \iV(T)-\{\rho_T\}$ by $T(v)$
the $L(v)$-tree with root $v$ obtained from $T$ 
by deleting the parent edge of $v$.

\begin{lemma}\label{lem:insights}
Suppose $T$ is a non-degenerate $X$-tree, 
$\sigma$ is a total ordering of $X$, and
$f:cl(T)\to X$ is a cluster
marker map for $T$ and $\sigma$. Then the following hold
\begin{enumerate}
\item[(i)]
$\cL_{(T,f)}$ is a minimal topological lasso for $T$.
\item[(ii)] $\Gamma(\cL_{(T,f)})$ is connected.
\item[(iii)] If $v$ and $w$ are distinct interior vertices of $T$
then $|\bigcup \cL_{(T,f)}(v)\cap \bigcup \cL_{(T,f)}(w)|\leq 1$.
\item[(iv)] Suppose $x\in X$. Then there exist distinct vertices
$v,w\in \iV(T)$ 
such that $x\in \bigcup \cL_{(T,f)}(v)\cap \bigcup \cL_{(T,f)}(w)$ 
if and only if there
exists some $u\in \iV(T)-\{\rho_T\}$ such that 
$x=f(L(u))$. 
\end{enumerate}
\end{lemma}
\begin{proof}
For all $v\in \iV(T)$, set $\cL(v)=\cL_{(T,f)}(v)$.

(i) This is an immediate consequence of 
Theorem~\ref{theo:characterization-topology} and
the respective definitions of the set
$\cL(v)$ where $v\in \iV(T)$
and the graph $G(\cL',v)$ where $\cL'$ is a set
of cords of $X$ and $v$ is again an interior vertex of $T$.

(ii) This is an immediate consequence of 
Proposition~\ref{prop:gamma-l-connected} combined with
Lemma~\ref{lem:insights}(i).

(iii) This is an immediate consequence of the fact that, for all vertices
$u\in \iV(T)$ and all $x,y\in \bigcup\cL(u)$ distinct, we have
$u=lca_T(x,y)$.

(iv) Let $x\in X$ and assume first that there exist 
distinct vertices $v,w\in \iV(T)$ 
such that $x\in \bigcup \cL(v)\cap \bigcup \cL(w)$ but
 $x\not =f(L_T(u))$, for all $u\in \iV(T)-\{\rho_T\}$. Then
 $x$ must be a leaf of $T$ that is simultaneously adjacent with $v$ and $w$
which is impossible. Thus, there must exist some $u\in \iV(T)$ such
that $x=f(L(u))$.

Conversely, assume that $x=f(L(u))$ for some 
$u\in \iV(T)-\{\rho_T\}$. Then $x\in L(u)$
and so there must exist an
interior vertex $w$ of $T(u)$ that is adjacent with $x$.
Hence, $x\in\bigcup \cL(w)$.  Let $v$ denote the
parent of $u$ in $T$ which exists since $u\not=\rho_T$. 
Then $x=f(L(u))\in \bigcup \cL(v)$ and so  
$x\in \bigcup \cL(v)\cap \bigcup \cL(w)$,
as required.
\qquad \end{proof}

Note that $u\in \{v,w\}$ need not hold for $u$, $v$
and $w$ as in the statement of Lemma~\ref{lem:insights}(iv). 
Indeed, suppose $T$ is the $X=\{a,b,c,d\}$-tree
with unique cherry $\{a,b\}$ and $d$ adjacent with the
root $\rho_T$ of $T$. Let $\sigma$ denote the lexicographic ordering 
of $X$ and let $f:cl(T)\to X$ be (the unique) 
cluster marker map for $T$ and $\sigma$.
Set $x=b$, $v=lca_T(a,b)$, $w=\rho_T$.
Then $x=f(L(u))$ where $u=lca_T(a,c)$ and
$x\in \bigcup\cL(v)\cap \bigcup\cL(w)$ but $u\not\in \{v,w\}$.

\begin{proposition}\label{prop:block}
Suppose $T$ is a non-degenerate $X$-tree,
 $\sigma$ is a total ordering of $X$, 
 and $f:cl(T)\to X$ is a cluster
marker map for $T$ and $\sigma$. Then  $\Gamma(\cL_{(T,f)})$
is a connected block graph and
every block of $\Gamma(\cL_{(T,f)})$ is of the form
$\Gamma(\cL_{(T,f)}(v))$, for some $v\in \iV(T)$.
\end{proposition}
\begin{proof}
For all $v\in \iV(T)$, set $\cL(v)=\cL_{(T,f)}(v)$ and
put $\cL=\cL_{(T,f)}$.
We claim that if $C$ is a cycle in $\Gamma(\cL)$ of length at least
three then there must exist some $v\in \iV(T)$ such that $C$ is
contained in $\Gamma(\cL(v))$. Assume to the contrary that this is not the
case, that is, there exists some cycle $C:u_1,u_2,\ldots, u_l,u_{l+1}=u_1$,
$l\geq 3$, in $\Gamma(\cL)$ such that, for all $v\in \iV(T)$,
we have that $C$ is not a cycle in $\Gamma(\cL(v))$. Without loss of
generality, we may assume that $C$ is of minimal length. For 
all $i\in\langle l\rangle$, put $v_i=lca_T(u_i,u_{i+1})$. Then, by 
the construction
of $\Gamma(\cL)$, we have for all such $i$ that $u_iu_{i+1}$ is 
an edge in $\Gamma(\cL(v_i))$ and, by the minimality of $C$,
that $v_i\not=v_j$ for all 
$i,j\in\langle l\rangle$
distinct. Put $Y=V(C)$ and let $T'=T|_Y$ denote the $Y$-tree obtained
by restricting $T$ to $Y$. Note that $lca_T(u_i,u_{i+1})=lca_{T'}(u_i,u_{i+1})$ 
holds for all $i\in\langle l\rangle$.
Thus, the  map
$\phi:E(C)\to \iV(T')$ defined by putting
 $u_iu_{i+1}\mapsto lca_{T}(u_i,u_{i+1})$,
$i\in\langle l \rangle$, is well-defined. 
Since $|E(C)|=l$ and
for any finite set $Z$ with three or more elements a $Z$-tree has
at most $|Z|-1$ interior vertices, it follows that there exist
$i,j\in \langle l \rangle$ distinct such that 
$\phi(u_i,u_{i+1})=\phi(u_j,u_{j+1})$. Consequently,
$v_i=lca_T(u_i,u_{i+1})=lca_T(u_j,u_{j+1})=v_j$
which is impossible and thus proves the claim.
Combined with Lemma~\ref{lem:insights}(ii) and (iii), 
it follows that $\Gamma(\cL)$ is a
connected block graph. That the blocks of $\Gamma(\cL)$ are of the
required form is an immediate consequence of the construction of
$\Gamma(\cL)$.
\qquad \end{proof}

To be able to establish that $\cL_{(T,f)}(v)$
is indeed a distinguished minimal topological lasso for $T$ and $f$
as above, we require a further concept. Suppose $A, B\subseteq X$ are
two distinct non-empty subsets of $X$. Then $A$ and $B$ are said to be
 {\em compatible} if $A\cap B\in\{\emptyset, A,B\}$. As is
 well-known (see e.\,g.\,\cite{DHKMS11,SS03}), for any $X$-tree $T'$ and any
  two vertices $v,w\in V(T')$
 the subsets $L(v)$ and $L(w)$ of $X$ are 
 compatible.

\begin{theorem}\label{theo: distinguished-lasso-verification}
Suppose $T$ is a non-degenerate  $X$-tree, 
$\sigma$ is a total ordering of $X$ and $f:cl(T)\to X$ is a cluster
marker map for $T$ and $\sigma$. Then 
$\cL_{(T,f)}$ is a distinguished minimal topological lasso for $T$.
\end{theorem}
 \begin{proof}
For all $v\in \iV(T)$ put $\cL(v) = \cL_{(T,f)}(v)$ and
put $\cL=\cL_{(T,f)}$.
In view of Proposition~\ref{prop:block} and Lemma~\ref{lem:insights}(i), 
it suffices to show that $\Gamma(\cL)$ is claw-free.
 Assume to the contrary that this is not the case and that there
exists some $x\in X$ that is contained in the vertex set of $m\geq 3$
blocks $A_1,\ldots,A_m$ of $\Gamma(\cL)$. Then, by 
Proposition~\ref{prop:block},
 there  exist distinct interior vertices $v_1, \ldots, v_m$ of $T$ such that,
 for all $i\in\langle m\rangle$, we have 
$V(A_i)=\bigcup\cL(v_i)\subseteq L(v_i)$. 
Since for all
$v,w\in V(T)$ distinct, the sets $L(v)$ and $L(w)$ are compatible,  
it follows that there exists a path $P$ from
$\rho_T$ to $x$ that contains the vertices $v_1,\ldots, v_m$ in its
vertex set. Without loss of generality we may assume that
$m=3$ and that, starting at $\rho_T$ and moving along $P$ 
the vertex $v_1$ is encountered first
then $v_2$ and then $v_3$. Note that 
$cl_{L(v_i)}(T)\not=\emptyset$, for $i=1,2$.
Since $T$ is a tree and so $x$ can neither be adjacent with 
$v_1$ nor with $v_2$ it follows that 
there must exist for $i=1,2$  some $B_i\in cl_{L(v_i)}(T)$ 
such that $x=f(B_i)$. But this is impossible as 
$B_2\in cl_{L(v_1)}(T)$ and so $f(B_1)\not=f(B_2)$ 
as $f$ is a cluster marker map for $T$ and $\sigma$.
\qquad \end{proof}

\section{Characterizing distinguished minimal topological 
lassos}\label{sec:characterization-distinguished}

In this section, we establish the converse of 
Theorem~\ref{theo: distinguished-lasso-verification}
which allows us to characterize 
distinguished minimal topological lasso of non-degenerate $X$-trees. 
We start with a well-known construction
for associating  an unrooted tree to a 
connected block graph (see e.\,g.\,\cite{diestel}).
Suppose that $G$ is a connected block graph. Then
we denote by $T_G$ the (unrooted) tree associated to $G$
with vertex set $Cut(G)\cup Block(G)$
and whose edges are of the from $\{a,B\}$ where $a\in Cut(G)$,
$B\in Block(G)$ and $a\in B$. Note that
if a vertex
$v\in V(T_G)$ is a leaf of $T_G$ then $v\in Block(G)$.

Suppose $T$ is a non-degenerate $X$-tree and
$\cL$ is a distinguished minimal topological lasso for $T$. 
Let $v$ denote an interior vertex of $T$ whose children are 
$v_1\ldots,v_l$ where $l=|ch(v)|$. Then
 Corollary~\ref{cor:bijection} combined with 
Proposition~\ref{prop:x-i-unique} implies that for all $i\in\langle l\rangle$
there exists a unique leaf $x_i\in L(v_i)$ of $T$ such that, 
for all $i,j\in\langle l\rangle$ distinct,
 $x_ix_j\in \cL$ and 
$\{x_1,\ldots, x_l\}=V(B_v)$. Since $\Gamma(\cL)$ is claw-free, 
every vertex of $B_v$  is contained in at most one further block
of $\Gamma(\cL)$. Thus, if $w\in V(B_v)$ and
$w\in V(B) $ holds too for some block $B\in Block(\Gamma(\cL))$
distinct from $B_v$
then $w$ must be a cut vertex of $\Gamma(\cL)$. For every 
vertex $v'\in \iV(T)$ that is the
child of some vertex $v\in \iV(T)$, we denote the  
unique element $x\in L(v')$ contained
in $V(B_v)$ by $c_{B_{v'}}$ in case $x\in Cut(\Gamma(\cL))$.
Note that it is not difficult to observe that,  in the tree 
$T_{\Gamma(\cL)}$, the vertex $c_{B_{v'}}$ is the vertex 
adjacent with $B_v$ that lies on the path 
from $B_v$ to $B_{v'}$. 

The following result lies at the heart of 
Theorem~\ref{theo:characterization} and establishes a 
crucial relationship between the non-root interior vertices of $T$ 
and the cut vertices of $\Gamma(\cL)$.

\begin{lemma}\label{lem:bijection-theta}
Suppose $T$ is an $X$-tree and $\cL$ is a distinguished minimal 
topological lasso
for $T$. Then  the map 
$$
\theta :\iV(T)-\{\rho_T\} \to Cut(\Gamma(\cL)):\,\,\,v\mapsto c_{B_{v}}
$$ 
is bijective. 
\end{lemma}
\begin{proof}
Clearly, $\theta$ is well-defined and injective. To see that $\theta$ is 
bijective
let $ T_{\Gamma(\cL)}^-$ denote the tree obtained from
$ T_{\Gamma(\cL)}$ by suppressing all degree two vertices. Then 
$Block(\Gamma(\cL))=V(T_{\Gamma(\cL)}^-)$ and 
Corollary~\ref{cor:bijection} implies that
$|Block(\Gamma(\cL))|=|\iV(T)|$ 
as $\Gamma(\cL)$ is a block graph. Since
$\Gamma(\cL)$ is claw-free, we clearly also have 
$|Cut(\Gamma(\cL))|=|E(T_{\Gamma(\cL)}^-)|$. Combined with the
fact that f $|V(T')|= |E(T')|+1$ holds for every tree $T'$,
it follows that 
 $|Cut(\Gamma(\cL))|=|Block(\Gamma(\cL))|-1=|\iV(T)|-1=
 |\iV(T)-\{\rho_T\}|$. Thus,
$\theta$ is bijective.
\qquad
\end{proof}

Armed with this result, we are now ready to establish the converse of 
Theorem~\ref{theo: distinguished-lasso-verification} which yields the
aforementioned characterization of distinguished
minimal topological lassos of non-degenerate $X$-trees.

\begin{theorem}\label{theo:characterization}
Suppose $T$ is a non-degenerate $X$-tree 
and $\cL$ is a set of cords of $X$. Then 
$\cL$ is a distinguished minimal topological lasso
for $T$ if and only if there exists a total ordering $\sigma$ of $X$ and
a cluster marker map $f$ for $T$ and $\sigma$ 
such that $\cL_{(T,f)}=\cL$.
\end{theorem}
 \begin{proof}
Assume first that $\sigma$ is some total ordering of $X$ and that 
$f:cl(T)\to X$ is a cluster marker map for  $T$ and $\sigma$. Then, by
Theorem~\ref{theo: distinguished-lasso-verification}, $\cL_{(T,f)}$
is a distinguished minimal topological lasso for $T$.

Conversely assume that $\cL$ is a distinguished minimal topological lasso
for $T$ and consider an embedding of $T$ into the plane. 
By abuse of terminology,  we will refer to this embedding of $T$ also as $T$.
We start with defining a total ordering $\sigma$ of $X$.
To this end, we first define a map $t:\iV(T)-\{\rho_T\}\to \mathbb N$
by setting, for all $v\in \iV(T)-\{\rho_T\}$, 
 $t(v)$ to be the length of the path from
$\rho_T$ and $v$. Put $h=\max\{t(v)\,:\, v\in \iV(T)-\{\rho_T\}\}$ 
and note that $h\geq 1$ as $T$ is non-degenerate. 
Starting at the left most interior vertex $v$ of $T$ for which $t(v)=h$ 
holds and moving, for all $l \in \langle h\rangle$, from left to right, 
we enumerate
all interior vertices of $T$ but the root. We next put $n=|X|$ and
$X=\langle n\rangle$
and relabel the elements in $X$ such that when traversing the circular ordering
induced by $T$ on $X\cup\{\rho_T\}$ in a counter-clockwise fashion we have
$\rho_T,1,2,3,\ldots, n,\rho_T$. To reflect this with regards to
$\cL$, we relabel the elements of the
cords in $\cL$ accordingly
and denote the resulting distinguished minimal topological lasso for
$T$ also by $\cL$.

By Lemma~\ref{lem:bijection-theta},
the map $\theta :\iV(T)-\{\rho_T\} \to Cut(\Gamma(\cL))$ 
defined in that lemma is bijective. Put $m=|Cut(\Gamma(\cL))|$ and
let $v_1,v_2,\ldots, v_m$ denote the enumeration of
the vertices in $\iV(T)-\{\rho_T\}$ obtained above. Also, 
set $Y=X-\{\theta(v_i): i\in \langle m\rangle\}$. Let 
 $y_1,y_2,\ldots, y_l$ denote an arbitrary but fixed total
 ordering of the elements 
of $Y$ where $l=|Y|$. Then we define $\sigma$ to be the total ordering of $X$
given by 
$$
\sigma:\,\, \theta(v_1),\theta(v_2),\ldots, \theta(v_{i-1}),
\theta(v_i),\theta(v_{i+1}),
,\ldots, \theta(v_m),y_1,y_2,\ldots, y_l
$$ 
where $\theta(v_1)$ is the minimal element and $y_l$ is the maximal element. 
Note that if $v\in \iV(T)$ is the 
parent of a pseudo-cherry $C$ of $T$ then $\theta(v)=\min_{\sigma}C$.

We briefly interrupt the proof of the theorem 
to illustrate these definitions
by means of an example. Put $X=\langle 13\rangle$ and consider the $X$-tree
$T$ depicted in  Fig.~\ref{fig:illustration-main-theorem}(i) (ignoring the
labelling of the interior vertices for the moment) and the distinguished
minimal topological lasso $\cL$ for $T$ pictured in the form of 
$\Gamma(\cL)$ in Fig.~\ref{fig:illustration-main-theorem}(ii). Then
the labelling of the interior vertices of $T$ gives the 
enumeration of those vertices considered in the
proof of Theorem~\ref{theo:characterization}. The total ordering $\sigma$
of $X$ restricted to the elements in $\{\theta(v_1),\ldots, \theta(v_6)\}$
is $3,5,12,1,10,7$. 
\begin{figure}[h]
\begin{center}
\input{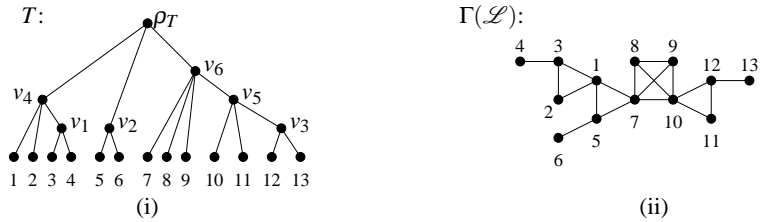}
\end{center}
\caption{\label{fig:illustration-main-theorem}
For  $X=\langle 13\rangle$ and the depicted $X$-tree $T$,
the enumeration of the interior vertices of $T$ 
considered in the proof of Theorem~\ref{theo:characterization} is indicated
in (i). With regards to this enumeration and the distinguished minimal 
topological lasso $\cL$ for $T$ pictured in the form of $\Gamma(\cL)$
in (ii), the total ordering $\sigma$ of $X$ considered in
that proof restricted to the elements in $\{\theta(v_1),\ldots, \theta(v_6)\}$
is $3,5,12,1,10,7$.
}
\end{figure}

Returning to the proof of the theorem, 
we claim that the map $f:cl(T)\to X$ given,
for all $A\in cl(T)$, by setting $f(A)=\theta(lca(A))$ is a cluster
marker map for $T$ and $\sigma$ where for all such $A$ we
put $lca(A)=lca_T(A)$. Indeed, suppose $A\in cl(T)$.
Then $\theta(lca(A))=c_{B_{lca(A)}}\in L(lca(A))$ holds by construction.
We distinguish between the cases that $cl_A(T)\not =\emptyset$ and that
$cl_A(T)=\emptyset$. If $cl_A(T)\not =\emptyset$
then since $\theta$ is bijective it follows that
 $\theta(lca(A))\not=\theta(v)$ holds for all descendants
 $v\in \iV(T)$ of $lca(A)$.
Combined with the definition of $\sigma$, we obtain  
$f(A)=\theta(lca(A))=\min_{\sigma}(A-\{\theta(lca(D))\,:\, D\in cl_A(T)\})=
\min_{\sigma}(A-\{f(D)\,:\, D\in cl_A(T)\})$, as required. 
If $cl_A(T)=\emptyset$ then, as was observed above,
 $f(A)=\theta(lca(A))=\min_{\sigma}A$. Thus,
$f$ is a cluster marker map for $T$ and $\sigma$, as claimed.

It remains to show that $\cL_{(T,f)}=\cL$. To see this note first
that, by Theorem~\ref{theo: distinguished-lasso-verification}, 
$\cL_{(T,f)}$ is a distinguished minimal topological lasso for $T$. Since
Lemma~\ref{lem:size-A(v)} implies that any
two minimal topological lasso for $T$ must be of
the same size and thus $|\cL_{(T,f)}|=|\cL|$ holds, it therefore
suffices to show that $\cL\subseteq \cL_{(T,f)}$. Suppose $a,b\in X$ distinct
 such that $ab\in \cL$.
Then there exists some interior vertex $v\in \iV(T)$ such that $v=lca_T(a,b)$.
Hence, $a,b\in V(B_v)$. We claim that $ab\in \cL_{(T,f)}(v)$.
To establish this claim,
we  distinguish between the cases that (i) $a\in ch(v)$ and (ii) that
$a\not\in ch(v)$.

Assume first that Case (i) holds, that is, $a$ is a child of $v$.
If $b\in ch(v)$ then the claim is an immediate consequence 
of the definition of $\cL_{(T,f)}(v)$. So assume that 
$b\not\in ch(v)$. Let $v'\in \iV(T)$ denote the child of $v$ for which
$b\in L(v)$ holds. Then $b=c_{B_{v'}}=\theta(v')=f(L(v'))$ follows by the 
observation preceding Lemma~\ref{lem:bijection-theta} combined 
with the fact that $b\in V(B_v)$. Hence,
$ab=af(L(v'))\in \cL_{(T,f)}(v)$, as claimed.

Assume next that Case (ii) holds, that is, $a$ is not a child of $v$.
In view of the previous subcase it suffices to consider the case that 
$b\not\in ch(v)$. Let $v',v''\in \iV(T)$ denote the
children of $v$ such that $a\in L(v')$ and $b\in L(v'')$. 
Then, again by the 
observation preceding Lemma~\ref{lem:bijection-theta} 
combined with the fact that $a,b\in V(B_v)$, we have
$a=c_{B_{v'}}=\theta(v')=f(L(v'))$ 
and $b=c_{B_{v''}}=\theta(v'')=f(L(v''))$ and so
$ab=f(L(v'))f(L(v''))\in \cL_{(T,f)}(v)$ follows, as claimed.
 This concludes the proof
of the claim and thus the proof of the theorem.
\qquad 
\end{proof}

We now take a brief break from our study of distinguished minimal 
topological lassos to point out a sufficient condition
for a set of cords to be a strong lasso for some 
$X$-tree which is implied
by Theorem~\ref{theo:characterization}. To make this
more precise, we need to introduce some more terminology
from \cite{HP13}. Suppose $T$ is an $X$-tree and  
$\cL$ is a set of cords of $X$.
Then $\cL$ is called an {\em equidistant lasso} for $T$ if, for all 
equidistant, proper edge-weightings $\omega$ and $\omega'$ of
$T$, we have that $\omega=\omega'$ holds whenever
$(T,\omega)$ and $(T,\omega')$ are $\cL$-isometric. Moreover, $\cL$
is called a {\em strong lasso} for $T$ if $\cL$  is simultaneously
an equidistant and a topological lasso for $T$ (see 
\cite{DHS11} for more
on such lassos in the unrooted case). 

Like a topological lasso for a $X$-tree $T$,
an equidistant lasso $\cL$ for $T$ 
can also be characterized in terms of a property of the 
child-edge graph $G(\cL,v)$ associated to $T$ and $\cL$
where $v\in \iV(T)$. Namely, a set 
$\cL$ of cords of $X$ is an equidistant lasso for an $X$-tree $T$
if and only if, for every vertex $v\in \iV(T)$, the graph
$G(\cL,v)$ has at least one edge (see \cite[Theorem 6.1]{HP13}).
Since for $\sigma$ some total ordering of $X$ and
$f:\iV(T)-\{\rho_T\}\to X$ a cluster marker map for $T$ and $\sigma$
the graphs $G(\cL_{(T,f)},v)$ clearly satisfy this property
for all $v\in \iV(T)$, it follows that $\cL_{(T,f)}$ is also
an equidistant lasso for $T$ and thus a strong lasso for $T$. Defining 
a strong lasso $\cL$ of an $X$-tree  to be {\em minimal} in analogy 
to when a topological lasso is minimal, 
Theorem~\ref{theo:characterization} implies

\begin{corollary}\label{corollary:strong-lasso-characterization}
Suppose $T$ is a non-degenerate $X$-tree, 
$\cL$ is a set of cords of $X$, $\sigma$
is a total ordering  of $X$, and  $f:cl(T)\to X$
is a cluster marker map for $T$ and $\sigma$. 
Then $\cL_{(T,f)}$ is a minimal strong lasso for $T$.
\end{corollary}

\section{Heredity of distinguished minimal 
topological lassos}\label{sec:subtree}

In this section, we turn our attention to the problems of characterizing when
a distinguished minimal topological lasso of an $X$-tree $T$ induces
a distinguished minimal topological lasso for a subtree of $T$ and, conversely
when distinguished minimal topological lassos of $X$-trees can be
combined to form a distinguished minimal topological lasso of a supertree
for those trees (see e.\,g.\,\cite{BE00} for more on such trees). This 
will also allow us to partially answer the rooted analogon of
a question raised in \cite{DHS11} for supertrees within the unrooted framework.
To make this more precise, we require further terminology.
Suppose $\cL$ a set of cords of $X$ and $Y\subseteq X$ is a 
non-empty subset. Then we set 
$$
\cL|_Y=\{ab\in\cL\,:\, a,b\in Y\}.
$$
Clearly, $\Gamma(\cL|_Y)$ is the
subgraph of $\Gamma(\cL)$ induced by $Y$ but
 $Y=\bigcup \cL|_Y$ need not hold. Moreover, if 
 $\cL$ is a minimal topological lasso for 
an $X$-tree $T$ and $|Y|\geq 3$ such that every interior vertex
of $T$ is also an interior vertex of $T|_Y$ then
 Theorem~\ref{theo:characterization-topology} implies that
$\cL|_Y$ is a minimal topological lasso for $T|_Y$. In particular, 
$\Gamma(\cL|_Y)$ must be connected in this case. 
The next result is a strengthening
of this observation. 

\begin{theorem}\label{theo:subtree}
Suppose $T$ is an $X$-tree, $\cL$ is a distinguished minimal
topological lasso for $T$, and $Y\subseteq X$ is a subset of size 
at least three. Then  $\cL|_Y$ is a distinguished
minimal topological lasso for $T|_Y$ if and only if $\Gamma(\cL|_Y)$
is connected.
\end{theorem}
\begin{proof}
Assume first that $\cL|_Y$ is a distinguished minimal topological
lasso for $T|_Y$. Then, by Proposition~\ref{prop:gamma-l-connected},
$\Gamma(\cL|_Y)$ is connected.

Conversely, assume that $\Gamma(\cL|_Y)$ is connected. Then the statement
clearly holds if $T$ is the star tree on $X$. So assume that $T$ is 
non-degenerate. Let $Y\subseteq X$ be of size at least three
and assume first that $T|_Y$ is the star tree on $Y$. We claim that 
 $\Gamma(\cL|_Y)$ is a clique. Assume to the contrary that
this is not the case, that is, there exist elements
$y,y'\in Y$ distinct such that $yy'\not\in \cL$. Since $\Gamma(\cL|_Y)$
is connected, there must exist a path $P:x_1=y,x_2,\ldots,x_l=y'$,
$l\geq 2$, in $\Gamma(\cL|_Y)$ from $y$ to $y'$. Since 
the vertex set of $\Gamma(\cL|_Y)$ is $Y$, it follows that
$X'=\{x_1,x_2,\ldots,x_l\}\subseteq Y$. Combined with the 
fact that $lca_T(x,x')=lca_T(Y)$ holds for all $x,x'\in X'$ distinct
as $T|_Y$ is a star tree on $Y$, we obtain $X'\subseteq V(B_{lca_T(Y)})$.
Thus, $yy'\in \cL$ which is impossible and thus proves the claim.
That $\cL|_Y$ is a distinguished minimal topological lasso for $T|_Y$
is a trivial consequence.

So assume that $T|_Y$ is
non-degenerate. Since $\cL$ is a distinguished minimal topological
lasso for $T$, 
Theorem~\ref{theo:characterization} implies that 
there exists a total ordering $\omega$ of $X$ and a cluster
marker map $f_{\omega}: cl(T)\to X$ for $T$ and $\omega$ such
that $\cL=\cL_{(T,f_{\omega})}$. Moreover, Lemma~\ref{lem:insights}(iv) 
implies that the cut-vertices
of $\Gamma(\cL)$ are of the form $f_{\omega}(L_T(v))$ where $v\in \iV(T)$.

To see that $\cL|_Y$ is a distinguished minimal topological lasso
for $T|_Y$ and some total ordering of $Y$ note first that 
the restriction $\sigma$ of $\omega$ to $Y$ induces 
a total ordering of $Y$.
Furthermore, the aforementioned form
of the cut-vertices of $\Gamma(\cL)$ 
combined with the assumption that $\Gamma(\cL|_Y)$ is connected
implies that, for all $A\in cl(T)$ with 
$A\cap Y\not=\emptyset$, we must have
$f_{\omega}(A)\in Y$. For all $A\in cl(T|_Y)$ denote by $A^T$ the 
set-inclusion minimal superset of $A$ contained in $cl(T)$. 
Then since $f_{\omega}$ is a cluster marker map for $T$
and $\omega$ it follows that
the map 
$$
f_{\sigma}:cl(T|_Y)\to Y\,:\, A\mapsto f_{\omega}(A^T)
$$
 is a cluster marker map for $T|_Y$ and $\sigma$.
By Theorem~\ref{theo:characterization} it now suffices to establish that
$\cL|_Y=\cL_{(T|_Y,f_{\sigma})}$. Since both $\cL|_Y$ and
$\cL_{(T|_Y,f_{\sigma})}$ are minimal topological lassos for 
$T|_Y$ and so $|\cL|_Y|=|\cL_{(T|_Y,f_{\sigma})}|$ is implied 
by Lemma~\ref{lem:size-A(v)} it suffices
to show that $\cL|_Y\subseteq \cL_{(T|_Y,f_{\sigma})}$.

Suppose $ab\in \cL|_Y$, that is,  $ab\in\cL$ and $a,b\in Y$. Since
$Y$ is the leaf set of $T|_Y$, there must exist a vertex $v\in \iV(T|_Y)$
such that $v=lca_{T|_Y}(a,b)$. Clearly, $v\in \iV(T)$. If $a$ and $b$
are both adjacent with $v$ in $T$ then $a$ and $b$ are also adjacent with $v$
in $T|_Y$. Thus $ab\in\cL_{(T|_Y,f_{\sigma})}(v)$ in this case.
So assume that at least one of $a$ and $b$ is not adjacent with  $v$ in $T$.
Without loss of generality  let $a$ denote that vertex.
Then since $ab\in \cL= \cL_{(T,f_{\omega})}$, it follows that there 
must exist a unique child $v'\in \iV(T)$
of $v$ such that $a\in L_T(v')$ and $a=f_{\omega}(L_T(v'))$. Hence, $a\in V(B_v)$
and a cut-vertex of $\Gamma(\cL)$. 

We claim that $v'\in \iV(T|_Y)$. Assume for contradiction that 
$v'\not\in \iV(T|_Y)$. Then since $f_{\omega}$ is a 
cluster marker map for $T$ and and $\omega$, it follows that 
$a$ cannot be a cut vertex in $\Gamma(\cL|_Y)$.
Since $\Gamma(\cL)$ is a claw-free block graph, no edge in
the unique block $B'\in Block(\Gamma(\cL))-\{B_v\}$ that 
also contains $a$ in its vertex set can therefore be incident with $a$
in $\Gamma(\cL|_Y)$. Since $\Gamma(\cL|_Y)$ is assumed to be 
connected, to obtain the required contradiction it now 
suffices to show that there exists some $c\in Y\cap L_T(v')$
distinct from $a$
such that every path from $c$ to $b$ in $\Gamma(\cL)$ crosses $a$.
But this is a consequence of the facts that $v$ is not the parent of $a$ 
in $T|_Y$ and, implied by Proposition~\ref{prop:gamma-l-connected}, that
the subgraph $\Gamma_{v'}(\cL)$ of $\Gamma(\cL)$ induced
by $L_T(v')$ is the connected component of $\Gamma(\cL)$ containing
$a$ obtained from $\Gamma(\cL)$ by deleting all edges in $B_v$ that
are incident with $a$. This concludes the proof of the claim

To conclude the proof of the theorem, note that if 
 $b$ is adjacent with $v$ in $T|_Y$ then
$ab= f_{\omega}(L_T(v'))b=f_{\omega}((L_{T|_Y}(v'))^T)b=
f_{\sigma}(L_{T|_Y}(v'))b\in\cL_{(T|_Y,f_{\sigma})}(v)
\subseteq \cL_{(T|_Y,f_{\sigma})}$.  If $b$ is not adjacent with $v$ in $T|_Y$
then there exists a child $v''\in \iV(T)$ of $v$ such that 
$b=f_{\omega}(L_T(v''))$.
In view of the previous claim,  we have $v''\in \iV(T|_Y)$.
But now arguments similar to the ones used before imply that
$ab\in \cL_{(T|_Y,f_{\sigma})}(v)\subseteq \cL_{(T|_Y,f_{\sigma})}$.
\qquad
\end{proof}

We now turn our attention to supertrees which are formally defined as
follows. Suppose  $\mathcal T=\{T_1,\ldots, T_l\}$, $l\geq 1$, is a
set of $Y_i$-trees $T_i$ with $Y_i\subseteq X$ and $|Y_i|\geq 3$, 
$i\in\langle l\rangle$,
and $T$ is an $X$-tree. Then $T$ is a called a {\em supertree} 
of $\mathcal T $ if
$T$ displays every tree in $\mathcal T$
where we say that some $X$-tree $T$ {\em displays} some $Y$-tree
$T'$ for $Y\subseteq X$ with $|Y|\geq 3$ if
$T|_Y$ and $T'$ are equivalent. More precisely, we have
the following result which relies on the fact that  
in case $\cL$ is a distinguished
minimal topological lasso for a {\em binary} $X$-tree 
$T$, that is, every vertex of $T$ but the leaves has two children, 
 $\Gamma(\cL)$ must be a path. In particular, $\cL$ induces a 
total ordering of the elements in $X$ in this case.
For $Y\subseteq X$ a non-empty subset of $X$,
we denote the maximal and minimal element in $Y$ with regards
to that ordering by $\min_{\cL}(Y)$ and $\max_{\cL}(Y)$, respectively.

\begin{corollary}\label{cor:supertree}
Suppose $X'$ and $X''$ are two non-empty subsets 
of $X$ such that $X=X'\cup X''$ and  $X'\cap X''\not=\emptyset$
 and $T'$ and $T''$ are $X'$-trees and $X''$-tree, respectively. Suppose 
also that $\cL'$ and $\cL''$  are distinguished minimal topological 
lassos for $T'$ and $T''$, respectively, such that 
$\cL'|_{X'\cap X''} =\cL''|_{X'\cap X''}$ and $\Gamma(\cL''|_{X'\cap X''})$
is connected.
If $T$ is a binary $X$-tree that displays both $T'$ and $T''$ 
then $\cL=\cL'\cup\cL''$ is a distinguished minimal topological 
lasso for $T$ if and only if
$\min_{\cL'}(X'\cap X'')\in \{\min_{\cL'}(X'), \min_{\cL''}(X'')\}$
and $\max_{\cL'}(X'\cap X'')\in \{\max_{\cL'}(X'), \max_{\cL''}(X'')\}$.
\end{corollary}

Continuing with the assumptions of Corollary~\ref{cor:supertree},
we also have that if 
$\min_{\cL'}(X'\cap X'')\in \{\min_{\cL'}(X'), \min_{\cL''}(X'')\}$
and $\max_{\cL'}(X'\cap X'')\in \{\max_{\cL'}(X'), \max_{\cL''}(X'')\}$
holds then $\cL'\cup \cL''$ is a (minimal) strong lasso for $T$
as every minimal topological lasso for an $X$-tree is also
an equidistant lasso for that tree. However, not all strong 
lassos for $T$ are
of this form. An example for this is furnished for $X'=\{a,c,d\}$ 
and $X''=\{a,b,c\}$ by the $X'$-tree $T'$, the $X''$-tree 
$T''$ and the $X'\cup X''$-tree $T$ depicted in Fig.~\ref{fig:supertree}
along with the set $\cL'=\{cd\}$  and $\cL''=\{ab,bc\}$ of cords
of $X'$ and $X''$, respectively. Clearly,
$T$ is a supertree of $\{T',T''\}$ and $\cL=\cL'\cup\cL''$   
is a strong lasso for $T$ but $\cL'$ is not even an 
equidistant lasso for $T'$. Investigating further the interplay
between minimal topological lassos for $X$-trees and minimal topological
lassos for supertrees that display them might therefore be of interest.

\begin{figure}[h]
\begin{center}
\input{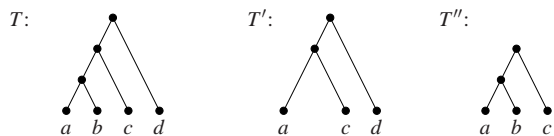}
\end{center}
\caption{\label{fig:supertree}
For $X'=\{a,c,d\}$ and $X''=\{a,b,c\}$ the $X'\cup X''$-tree $T$ is a supertree
for the depicted $X'$ and $X''$ trees $T'$ and $T''$, respectively. Clearly,
$\cL'=\{cd\}$ and $\cL''=\{ab,bc\}$ are sets of cords of $X'$ and $X''$,
respectively, and $\cL=\cL'\cup\cL''$ is a strong lasso for $T$ but
$\cL'$ is not even an equidistant lasso for $T'$.
}
\end{figure}

We conclude with returning to Fig.~\ref{fig:block-graph-motivation}
which depicts two non-equivalent $X$-trees that are topologically
lassoed by the same set $\cL$ of cords of $X$. In fact,
$\cL$ is even a minimal topological lasso for both of them. 
Understanding better the relationship between $X$-trees that are
topologically lassoed by the same set of cords of $X$ might also
be of interest to study further.

\bibliographystyle{spbasic} 
\bibliography{phylogenetics4}
\end{document}